\documentclass[reprint,amsmath,amssymb,aps,prd,superscriptaddress]{revtex4-2}

\usepackage{hyperref}
\usepackage{graphicx}
\usepackage{tikz-cd}

\hypersetup{
  colorlinks=true,
  linkcolor=blue,
  citecolor=blue,
  urlcolor=blue
}

\hypersetup{
  pdftitle={Induced Couplings and Causal Bounds from Nondegenerate Dirac Lagrangians},
  pdfauthor={M. Netz-Marzola, D. Vasak, V. Denk, J. Struckmeier, and H. St\"ocker}
}

\newtheorem{proposition}{Proposition}

\newcommand{\afffias}{Frankfurt Institute for Advanced Studies (FIAS), Ruth-Moufang-Str.~1, 60438 Frankfurt am Main, Germany}
\newcommand{\affgoethe}{Physics Department, Goethe University, Max-von-Laue-Str.~1, 60438 Frankfurt am Main, Germany}

\begin{document}

\title{Induced Couplings and Causal Bounds from Nondegenerate Dirac Lagrangians}

\author{M.~Netz-Marzola}
\email{marzola@fias.uni-frankfurt.de}
\affiliation{\afffias}
\affiliation{\affgoethe}
\author{D.~Vasak}
\email{vasak@fias.uni-frankfurt.de}
\affiliation{\afffias}
\author{V.~Denk}
\email{denk@fias.uni-frankfurt.de}
\affiliation{\afffias}
\affiliation{\affgoethe}
\author{J.~Struckmeier}
\email{struckmeier@fias.uni-frankfurt.de}
\affiliation{\afffias}
\affiliation{\affgoethe}
\author{H.~St\"ocker}
\email{stoecker@fias.uni-frankfurt.de}
\affiliation{\afffias}
\affiliation{\affgoethe}

\date{\today}

\begin{abstract}
The standard Dirac Lagrangian is linear in the field derivatives and therefore has a vanishing Hessian. We identify the minimal null deformations of this Lagrangian that make the covariant Legendre map
locally invertible. Imposing global phase invariance, reality, proper
Poincar\'e covariance, and absence of external background tensors leaves
the two-parameter spinorial bivector
$\mathcal E^{\mu\nu}=\ell\sigma^{\mu\nu}
+\ell_5{}^\star\sigma^{\mu\nu}$, where the star denotes the Hodge dual and
$\ell^2+\ell_5^2\neq0$. This extends the analysis of \cite{struckmeier2024pauli} by a parity-odd term. After minimal $U(1)$ gauging, this free null term is
no longer variationally trivial and induces magnetic- and
electric-dipole Pauli operators, together with an identically conserved
dipole current. These dipole terms make the species-dependent regularization lengths
directly constrainable by precision moment measurements. Metric-affine gauging in a Lorentz-spinor prescription
then produces spin-curvature, torsion, and nonmetricity corrections to
the Dirac operator. Applying the Velo--Zwanziger
criterion, we exclude all nonzero pure axial, pure trace, and mixed
axial--trace torsion backgrounds, as well as nonzero pure Weyl and
second-trace nonmetricity backgrounds. The combined trace-vector
nonmetricity sector is not excluded only when the effective trace vector vanishes. Tensor torsion and general tracefree
nonmetricity remain unclassified without further algebraic assumptions,
while the Levi--Civita limit preserves the metric light cone and leaves
only lower-order curvature-dependent mass terms.  Thus,
gauging a Legendre-regular Dirac representative turns a free variational
ambiguity into observable couplings, while non-Riemannian sectors are sharply constrained by causality bounds.
\end{abstract}

\maketitle

\section{Introduction}
\label{sec:introduction}

The standard Dirac Lagrangian is first order and linear in the spinor derivatives. This
feature is normally regarded as natural: it gives the Dirac equation directly
and leads, in canonical formulations, to the familiar primary constraints of
fermionic field theory. From the viewpoint of covariant Hamiltonian field
theory, however, the same feature has a different consequence. Since the
Lagrangian is less than quadratic in the spacetime derivatives of \(\psi\) and
\(\bar\psi\), the associated covariant Legendre map is singular. The
polymomenta conjugate to the spinor derivatives are functions of the fields
alone, not of the derivatives, and therefore the derivative variables cannot
be recovered from them. Thus the standard Dirac theory does not
directly define a regular covariant Hamilton, or De Donder--Weyl, system in
the ordinary sense \cite{dirac,henneaux92}.

This observation is not, by itself, a physical inconsistency. Singular
Lagrangians are common in relativistic field theory, and constrained
Hamiltonian methods provide the standard way to treat them. The point of the
present work is instead to exploit the variational ambiguity of the
Lagrangian. Adding a null Lagrangian, locally a total divergence, leaves the
Euler--Lagrange equations unchanged, but it can change the fiber Hessian and
therefore the covariant Legendre map. For the free Dirac field, this turns
Legendre regularization into a classification problem within the same free
variational class. The task is not merely to write down a regularizing
surface term, but to determine which first-order null terms are compatible
with the spinorial symmetries and actually remove the Legendre degeneracy.

This question is motivated in part by the non-degenerate Dirac representative
used in \cite{struckmeier2024pauli}. There, the Gasiorowicz
derivative-quadratic surface term renders the free Dirac Lagrangian
non-degenerate while preserving the free Dirac equation. After minimal 
coupling, the same term is no longer a divergence, because
gauge-covariant derivatives do not commute, and, in the electromagnetic case, the Pauli interaction is
generated. In this sense the length scale that regularizes the Legendre map
becomes a physical dipole coupling after gauging. The same idea was then
applied to curved spacetime, where Pauli-type spin-torsion and spin-curvature
terms arise.

The first purpose of the present paper is to replace this special choice by
a classification. We consider first-order, bilinear, charge-neutral, real,
proper Poincar\'e-covariant null deformations of the free Dirac Lagrangian,
without external background tensors, and ask which of them can make the
covariant Legendre map locally invertible. The result is the two-parameter
expression,
\[
  \partial_\mu\bar\psi\left(\ell\sigma^{\mu\nu}
  +
  \ell_5{}^\star\sigma^{\mu\nu}\right)\partial_\nu\psi,
\]
with \(\ell^2+\ell_5^2\neq0\), where a star denotes the Hodge dual. The Gasiorowicz term of
\cite{struckmeier2024pauli} is the parity-even member of this family. If
parity or charge-parity (CP) invariance is imposed, the second, dual branch is removed.

The second purpose is to study how the regularized Lagrangian interacts with gauge fields. Under minimal \(U(1)\) coupling, the parity-even parameter \(\ell\) induces the usual magnetic
Pauli-type operator, while the parity-odd parameter \(\ell_5\) induces its
electric-dipole-type partner.
For several Dirac species, these parameters should be read as
\((\ell_f,\ell_{5f})\), properties of the chosen
free representative of the fermion species, rather than coefficients
introduced specifically for a particular coupling. Electromagnetic moment measurements constrain the same
regularization lengths that also enter any other coupling fields.

Under spacetime gauging, we use a
Lorentz-spinor prescription. This yields a generalized
Dirac operator containing spin-curvature, torsion, and nonmetricity
corrections. In the Riemann--Cartan limit and in the parity-even branch, the
purely axial torsion sector reproduces
the Pauli-type coupling terms obtained from
gravitational dipole moments in nonminimal Poincar\'e gauge theory
\cite{obukhov2014spin}.

In non-Riemannian sectors, the gravitational coupling terms can
also modify the principal part of the Dirac equation and hence its
characteristic structure. We analyze this using the Velo--Zwanziger
criterion, treating the non-Riemannian background locally as fixed.
This is the same diagnostic used in \cite{fabbri2019restrictions} to
restrict derivative torsion--spinor couplings, but here the relevant
terms are induced by Legendre regularization and minimal gauging. The
aim is to determine which torsion and nonmetricity components are
compatible with causal spinor front propagation for the induced
operator.

The Levi--Civita limit is qualitatively different, only lower-order scalar and pseudoscalar
curvature-dependent mass terms survive. Thus the torsionless,
metric-compatible sector (when
$\tilde{\mathcal L}_{\rm Gr}$ is chosen to be the Einstein--Hilbert
density) is not excluded by the
Velo--Zwanziger criterion. The restrictions derived below concern the
non-Riemannian sectors.

The results are summarized in Table~\ref{tab:VZ-summary}. In the
Riemann--Cartan axial--trace subsector, pure axial torsion, pure trace
torsion, and mixed axial--trace torsion all fail the Velo--Zwanziger test for every nonzero
background. This eliminates the causal windows reported in
\cite{struckmeier2024pauli} for the axial and vector torsion cases. Tensor
torsion is not excluded in general, since its characteristic determinant
does not reduce to a universal one-vector cone. In the torsionless
nonmetricity limit, the trace-vector sectors are similarly restricted, while
the tracefree sector remains open in general.

The resulting equations also connect naturally with SME-type fermion
operators. The general quadratic Dirac operator and its dispersion relation
are structured in \cite{kostelecky2013fermions} by decomposing all possible
corrections on the Clifford basis. At a locally frozen background, the
torsion and nonmetricity terms found here have the same algebraic character:
they enter as Clifford-valued coefficients in the Dirac operator. Similar
torsion--SME identifications, especially for constant torsion backgrounds,
were discussed in \cite{shapiro2002physical,kostelecky2008constraints}. In the
present work, however, these coefficients are fixed geometrically by the
Legendre-regularization scale and by the non-Riemannian fields.

Finally, the free regularization is closely related to the Lepagean
regularization of singular covariant Hamiltonian systems
\cite{krupkova2001legendre}. In that approach one does not necessarily add a
new term to the physical action. Instead, one modifies the
Poincar\'e--Cartan form by a \(2\)-contact term so that the Hamiltonian
Legendre map becomes regular while the Euler--Lagrange equations remain
unchanged. When this construction can be written as an equivalent ordinary
Lagrangian, the result is called a dedonderized Lagrangian. The
derivative-quadratic Dirac null term classified here has precisely this
structure at the free level.

This distinction suggests a possible alternative use of the same algebraic
bivector. In the route followed here, and in
\cite{struckmeier2024pauli}, the dedonderized Lagrangian is taken as the
matter action and then gauged. This produces the dipole terms, but also the
metric-affine principal-symbol obstructions analyzed below. A purely
Lepagean use would keep the correction invariant. That would not generate the Pauli-type
dipole couplings, but it may provide a regular covariant Hamiltonian fermion
sector without metric-affine propagation problems.

The paper is organized as follows. Section~\ref{sec:regular-dirac} reviews
covariant Legendre regularity, proves the singularity of the ordinary Dirac
Lagrangian, classifies the admissible null regularizations, and constructs
the corresponding covariant Hamiltonian. Section~\ref{subsec:qed} applies
minimal \(U(1)\) gauging and derives the induced dipole current, giving the corresponding parameter constraints
from the magnetic-moment and the electric-dipole bounds.
Section~\ref{sec:spacetime-metricaffine} gauges the theory to
metric-affine spacetime and derives the generalized Dirac operator in
the Riemann--Cartan, torsionless-nonmetric, and Levi--Civita limits.
Section~\ref{characteristics} computes the
characteristic determinant and applies the Velo--Zwanziger criterion to the
torsion and nonmetricity sectors summarized in Table~\ref{tab:VZ-summary}.
The appendices collect geometric intuition for fiberwise regularity, the
curvature--Clifford identities, and the comparison with the gravitational
dipole couplings of \cite{obukhov2014spin}.

Throughout the paper we use natural units \(\hbar=c=1\) and metric signature
\((+,-,-,-)\). The Dirac matrices satisfy
\[
  \{\gamma^\mu,\gamma^\nu\}=2\eta^{\mu\nu}\mathbf 1,
  \qquad
  \sigma^{\mu\nu}:=\frac{i}{2}[\gamma^\mu,\gamma^\nu],
\]
and
\[
  \gamma^5:=i\gamma^0\gamma^1\gamma^2\gamma^3,
  \qquad
  \varepsilon^{0123}=+1.
\]
With these conventions,
\[
  {}^\star\sigma^{\mu\nu}
  :=
  \frac12\varepsilon^{\mu\nu\rho\sigma}\sigma_{\rho\sigma}
  =
  -i\sigma^{\mu\nu}\gamma^5 .
\]
Spinor components are treated as classical field coordinates in the
covariant Legendre analysis, with \(\psi\) and \(\bar\psi\) varied
independently. Pure Levi--Civita terms are denoted by an overcircle.

\section{The regular Dirac Lagrangian}
\label{sec:regular-dirac}

\subsection{Regularity and the covariant Legendre map}
\label{subsec:regularity}

We consider an \(m\)-dimensional spacetime Minkowski manifold \(\mathcal{M}\), with
coordinates \(x^\mu\). Let
\begin{equation*}
  \phi^A=\phi^A(x)    
\end{equation*}
denote a collection of fields or field components. For a first-order Lagrangian theory,
we introduce independent derivative variables $v^A{}_\mu .$ Along an actual field configuration these are set equal to the derivatives of
the fields,
\begin{equation*}
  v^A{}_\mu=\partial_\mu \phi^A .
\end{equation*}

A first-order Lagrangian density is written locally as
\begin{equation*}
  L=\mathcal L(x^\mu,\phi^A,v^A{}_\mu)\,d^m x,
  \qquad
  d^m x=dx^1\wedge\cdots\wedge dx^m .
\end{equation*}

The covariant Legendre map associated with \(\mathcal L\) is the map
\begin{equation}
  \mathbb F\mathcal L:
  (x^\mu,\phi^A,v^A{}_\mu)
  \longmapsto
  (x^\mu,\phi^A,\pi_A{}^\mu),
  \label{eq:coordinate-legendre-map}
\end{equation}
where the covariant momenta are defined by
\begin{equation}
  \pi_A{}^\mu
  :=
  \frac{\partial\mathcal L}{\partial v^A{}_\mu}.
  \label{eq:covariant-momenta-coordinate}
\end{equation}
The Lagrangian is called regular if this map can be
locally inverted for the derivative variables \(v^A{}_\mu\). By the inverse
function theorem, this is equivalent to the nondegeneracy of the matrix
\begin{equation}
  W_{(A\mu)(B\nu)}
  =
  \frac{\partial \pi_A{}^\mu}{\partial v^B{}_\nu}
  =
  \frac{\partial^2\mathcal L}
       {\partial v^A{}_\mu\,\partial v^B{}_\nu}.
  \label{eq:fiber-hessian}
\end{equation}
Thus the fiberwise regularity condition is
\begin{equation}
  \det\!\left(W_{(A\mu)(B\nu)}\right)\neq0 .
  \label{eq:regularity-condition}
\end{equation}
Here the Hessian is viewed as a matrix whose rows and columns are indexed by
the combined labels \((A,\mu)\) and \((B,\nu)\).

Equivalently, regularity can be expressed infinitesimally. Let
\(\varepsilon\in\mathbb R\) be a small parameter and let
\(\delta v^A{}_\mu\) be an arbitrary variation of the derivative variables.
Under
\begin{equation*}
  v^A{}_\mu
  \mapsto
  v^A{}_\mu+\varepsilon\,\delta v^A{}_\mu ,
\end{equation*}
the induced first-order change of the momenta is
\begin{equation*}
  \delta \pi_A{}^\mu
  =
  W_{(A\mu)(B\nu)}\,\delta v^B{}_\nu .
\end{equation*}
Therefore regularity means that no nonzero infinitesimal change of the
derivative variables is invisible to the momenta:
\begin{equation*}
  W_{(A\mu)(B\nu)}\,\delta v^B{}_\nu=0
  \quad\Longrightarrow\quad
  \delta v^A{}_\mu=0 .
\end{equation*}

When the condition \eqref{eq:regularity-condition} holds, the derivative
variables can be solved locally as functions of the momenta,
\begin{equation*}
  v^A{}_\mu
  =
  v^A{}_\mu(x^\nu,\phi^B,\pi_B{}^\nu).
\end{equation*}
A local covariant Hamiltonian function may then be defined by
\begin{equation}
  \mathcal H(x,\phi,\pi)
  =
  \pi_A{}^\mu v^A{}_\mu(x,\phi,\pi)
  -
  \mathcal L\bigl(x,\phi,v(x,\phi,\pi)\bigr),
  \label{eq:coordinate-hamiltonian}
\end{equation}
with Hamiltonian density
\begin{equation*}
  H=\mathcal H\,d^m x .
\end{equation*}
If the Hessian is degenerate, the Legendre map cannot be invertible. Many physically relevant field theories are singular in this sense. In particular, any first-order Lagrangian that is merely linear in the derivative variables has a vanishing Hessian and therefore fails to define an invertible covariant Legendre map. The standard free Dirac Lagrangian is a famous instance of this type of degeneracy \cite{dirac,henneaux92}.

\subsection{Singularity of the standard free Dirac Lagrangian}
\label{subsec:free-dirac-singular}

We now specialize to a single free Dirac field in four-dimensional Minkowski
spacetime. The field variables are a Dirac spinor \(\psi\) and its Dirac
adjoint \(\bar\psi\), treated as independent variables in the variational
calculus. In the notation of the previous subsection,
\begin{equation*}
  \phi^A=(\psi,\bar\psi).
\end{equation*}
The corresponding independent derivative variables are
\begin{equation*}
  v^A{}_\mu=(\psi_\mu,\bar\psi_\mu),
\end{equation*}
which are set equal to the actual derivatives along a field configuration:
\begin{equation*}
  \psi_\mu=\partial_\mu\psi,
  \qquad
  \bar\psi_\mu=\partial_\mu\bar\psi .
\end{equation*}
Derivatives with respect to \(\psi\), \(\bar\psi\), \(\psi_\mu\), and
\(\bar\psi_\mu\) are understood componentwise, with \(\psi\) and
\(\bar\psi\) treated as independent field coordinates.

The standard symmetrized free Dirac Lagrangian density is
\begin{equation*}
  L_D=\mathcal L_D\,d^4x,
\end{equation*}
where, in terms of the independent derivative variables,
\begin{equation}
  \mathcal L_D
  =
  \frac{i}{2}
  \left(
    \bar\psi\gamma^\mu\psi_\mu
    -
    \bar\psi_\mu\gamma^\mu\psi
  \right)
  -
  m\bar\psi\psi .
  \label{eq:standard-dirac-lagrangian}
\end{equation}

The covariant momentum conjugate to \(\psi_\mu\) has the algebraic type of a
barred spinor and is denoted by \(\bar\pi^\mu\), while the covariant momentum
conjugate to \(\bar\psi_\mu\) has the algebraic type of an unbarred spinor
and is denoted by \(\pi^\mu\):
\begin{equation}
  \bar\pi^\mu
  :=
  \frac{\partial\mathcal L_D}{\partial\psi_\mu}
  =
  \frac{i}{2}\bar\psi\gamma^\mu,
  \qquad
  \pi^\mu
  :=
  \frac{\partial\mathcal L_D}{\partial\bar\psi_\mu}
  =
  -\frac{i}{2}\gamma^\mu\psi .
  \label{eq:dirac-polymomenta}
\end{equation}
Thus the covariant Legendre map associated with \(\mathcal L_D\) is
\begin{equation*}
  \mathbb F\mathcal L_D:
  (x^\mu,\psi,\bar\psi,\psi_\mu,\bar\psi_\mu)
  \longmapsto
  (x^\mu,\psi,\bar\psi,\bar\pi^\mu,\pi^\mu),
\end{equation*}
with
\begin{equation*}
  \bar\pi^\mu
  =
  \frac{i}{2}\bar\psi\gamma^\mu,
  \qquad
  \pi^\mu
  =
  -\frac{i}{2}\gamma^\mu\psi .
\end{equation*}
The right-hand side depends only on the fields \(\psi\) and \(\bar\psi\), not
on the derivative variables \(\psi_\mu\) and \(\bar\psi_\mu\). Hence the
Legendre map cannot be inverted to solve for the derivatives.

Equivalently, \(\mathcal L_D\) is linear in the derivative variables
\(\psi_\mu\) and \(\bar\psi_\mu\). Therefore all second derivatives of
\(\mathcal L_D\) with respect to these variables vanish:
\begin{equation*}
  \frac{\partial^2\mathcal L_D}
       {\partial\psi_\mu\,\partial\psi_\nu}
  =
  \frac{\partial^2\mathcal L_D}
       {\partial\bar\psi_\mu\,\partial\bar\psi_\nu}
  =
  \frac{\partial^2\mathcal L_D}
       {\partial\bar\psi_\mu\,\partial\psi_\nu}
  =
  \frac{\partial^2\mathcal L_D}
       {\partial\psi_\mu\,\partial\bar\psi_\nu}
  =
  0 .
\end{equation*}
Thus the full Hessian of the
Dirac Legendre map is identically zero:
\begin{equation*}
  W_D=0,
  \qquad
  \det W_D=0 .
\end{equation*}
The standard free Dirac Lagrangian is therefore singular in the covariant
Legendre sense.

Instead of determining the derivative variables from the momenta, the
Legendre map imposes the primary constraints
\begin{equation}
  \bar\pi^\mu
  -
  \frac{i}{2}\bar\psi\gamma^\mu
  =
  0,
  \qquad
  \pi^\mu
  +
  \frac{i}{2}\gamma^\mu\psi
  =
  0 .
  \label{eq:dirac-primary-relations}
\end{equation}

\subsection{Minimal null regularizations}
\label{subsec:null-regularizations}

We now ask whether this Legendre degeneracy
can be lifted without changing the free flat-space Dirac equation. Thus the
regularizing term must change the derivative-fiber Hessian, but it must be
variationally trivial in Minkowski space.

A regularization of the free Dirac Legendre map must therefore add a nonzero
quadratic contribution in the derivative variables. The geometric examples of
Appendix~\ref{app:visualizing-regularity} also show that this contribution
need not be positive or convex -- a saddle-type quadratic term is sufficient,
provided that its Hessian has no null directions.

At the same time, preserving the free Dirac equation requires the deformation
to be a null Lagrangian. We therefore seek a first-order local deformation
\[
  \mathcal L_D\longmapsto \mathcal L_D+\Delta\mathcal L
\]
which is locally a total divergence in Minkowski space,
\begin{equation*}
  \Delta\mathcal L=\partial_\mu F^\mu ,
\end{equation*}
for a local current \(F^\mu\)
\cite{Hojman1983,KrupkaMusilova1998,PaleseRossiWinterrothMusilova2016}.

We consider the standard structural requirements appropriate to the free Dirac system: the global phase symmetry
\begin{equation*}
  \psi\mapsto e^{i\alpha}\psi,
  \qquad
  \bar\psi\mapsto \bar\psi\,e^{-i\alpha},
\end{equation*}
translation invariance under \(\mathbb R^{1,3}\), standard
\(SL(2,\mathbb C)\)-covariance of the spinor sector, absence of external
background fields, and reality of the action.

Since the Hessian is only sensitive to the part of the Lagrangian which is at least quadratic in
\((\psi_\mu,\bar\psi_\mu)\), lower-derivative terms cannot contribute to
regularity.  Conversely, higher charge-neutral nonlinearities do not give a
minimal regularization of the free system, once the global phase symmetry
is required, any term beyond the bilinear Dirac sector contains at least two
unbarred and two barred spinorial factors.  Such a
term is at least quartic in
\[
  \psi,\bar\psi,\psi_\mu,\bar\psi_\mu .
\]
After two differentiations with respect to the derivative variables, its
contribution to the Hessian still vanishes at the free vacuum
\begin{equation*}
  (\psi,\bar\psi,\psi_\mu,\bar\psi_\mu)=(0,0,0,0) .
\end{equation*}
They therefore cannot provide a pointwise regularization of the free Dirac
Legendre map at the origin. The Hessian-relevant sector is the
quadratic Dirac sector.
Henceforth, we drop the $\psi_\mu$, $\bar\psi_\mu$ notation in favor of the more standard $\partial_\mu\psi$, $\partial_\mu\bar\psi$.

\begin{proposition}[Hessian-relevant representative]
\label{prop:hessian-relevant-null-sector}
Every first-order, charge-neutral, bilinear null deformation has, modulo null
terms with identically vanishing Hessian, a representative of the form
\begin{equation}
  \Delta\mathcal L
  =
  (\partial_\mu\bar\psi)\,E^{[\mu\nu]}\,(\partial_\nu\psi),
  \label{eq:DL-antisymmetric-general}
\end{equation}
where \(E^{[\mu\nu]}\) is a spacetime-independent, Dirac-matrix-valued
coefficient.
\end{proposition}

\noindent\textit{Proof.}
Although only derivative-quadratic terms can affect the Hessian,
variational triviality is a condition on the whole first-order bilinear
density.  We therefore begin with the full charge-neutral bilinear expression
and discard the Hessian-trivial part only after imposing the null
condition.  The most general such expression is
\begin{equation}
  \Delta\mathcal L
  =
  \bar\psi A\psi
  +
  \bar\psi B^\mu\partial_\mu\psi
  +
  (\partial_\mu\bar\psi)C^\mu\psi
  +
  (\partial_\mu\bar\psi)E^{\mu\nu}\partial_\nu\psi ,
  \label{eq:DL-general}
\end{equation}
where \(A\), \(B^\mu\), \(C^\mu\), and \(E^{\mu\nu}\) are constant
Dirac-algebra-valued coefficients carrying the displayed Lorentz indices.
The first three terms are at most linear in the derivative variables, and
therefore have null Hessian.

Variational triviality means that the Euler--Lagrange derivatives of
\(\Delta\mathcal L\) vanish identically. Taking the derivative with respect
to \(\bar\psi\) thus gives
\begin{align*}
  0
  =
  \frac{\delta\,\Delta\mathcal L}{\delta\bar\psi}
  &=
  \frac{\partial\Delta\mathcal L}{\partial\bar\psi}
  -
  \partial_\mu
  \left(
    \frac{\partial\Delta\mathcal L}
         {\partial(\partial_\mu\bar\psi)}
  \right)
  \nonumber\\
  &=
  A\psi
  +
  (B^\mu-C^\mu)\partial_\mu\psi
  -
  E^{\mu\nu}\partial_\mu\partial_\nu\psi .
\end{align*}
Since \(\partial_\mu\partial_\nu\psi\) is symmetric in \(\mu,\nu\), only the
symmetric part \(E^{(\mu\nu)}\) contributes. Therefore
\begin{equation}
  A=0,
  \qquad
  B^\mu=C^\mu,
  \qquad
  E^{(\mu\nu)}=0 .
  \label{eq:null-conditions}
\end{equation}
The Euler--Lagrange derivative with respect to \(\psi\) gives the same
conditions.

Substituting \eqref{eq:null-conditions} into \eqref{eq:DL-general} yields
\begin{equation*}
  \Delta\mathcal L
  =
  \partial_\mu(\bar\psi B^\mu\psi)
  +
  (\partial_\mu\bar\psi)E^{[\mu\nu]}\partial_\nu\psi ,
\end{equation*}
where the first term is now seen to be a null divergence with identically vanishing Hessian.  It is therefore precisely the kind of term discarded in the statement of the proposition.

The second term, however, is quadratic in the derivative variables. Since \(E^{[\mu\nu]}\) is antisymmetric and constant, 
\begin{equation*}
  (\partial_\mu\bar\psi)E^{[\mu\nu]}\partial_\nu\psi
  =
  \partial_\mu
  \left[
    \frac12
    \left(
      \bar\psi E^{[\mu\nu]}\partial_\nu\psi
      -
      (\partial_\nu\bar\psi)E^{[\mu\nu]}\psi
    \right)
  \right].
\end{equation*}
Thus it leaves the free Euler--Lagrange equations unchanged but may still
have a nonzero Hessian.
\hfill\(\square\)

\begin{proposition}[Proper Lorentz classification]
\label{prop:lorentz-classification}
Under proper Poincar\'e covariance, absence of external backgrounds, and
reality of the action, every Hessian-relevant null deformation of the form
\eqref{eq:DL-antisymmetric-general} is
\begin{equation}
  \Delta\mathcal L
  =
  \frac{i}{3}\,
  (\partial_\mu\bar\psi)\,
  \mathcal E^{\mu\nu}\,
  (\partial_\nu\psi),
  \label{eq:DL-final}
\end{equation}
with
\begin{equation}
  \mathcal E^{\mu\nu}
  =
  \ell\,\sigma^{\mu\nu}
  +
  \ell_5\, ^\star\sigma^{\mu\nu},
  \qquad
  \ell,\ell_5\in\mathbb R,
  \label{eq:E-general-final}
\end{equation}
where
\begin{equation}
   ^\star\sigma^{\mu\nu}
  :=
  \frac12\,\varepsilon^{\mu\nu\rho\sigma}\sigma_{\rho\sigma}.
  \label{eq:dual-sigma}
\end{equation}
The constants \(\ell\) and \(\ell_5\) have dimension of length, or inverse
mass.
\end{proposition}

\noindent\textit{Proof.}
By Proposition~\ref{prop:hessian-relevant-null-sector}, only $E^{[\mu\nu]}$ remains to be classified. 

Reality of the action requires:
\begin{align*}
  \left[
    (\partial_\mu\bar\psi)E^{[\mu\nu]}(\partial_\nu\psi)
  \right]^\dagger
  &=
  (\partial_\nu\bar\psi)\,
  \gamma^0\bigl(E^{[\mu\nu]}\bigr)^\dagger\gamma^0\,
  (\partial_\mu\psi)
  \nonumber\\
  &=
  -(\partial_\mu\bar\psi)\,
  \gamma^0\bigl(E^{[\mu\nu]}\bigr)^\dagger\gamma^0\,
  (\partial_\nu\psi),
\end{align*}
where the last step relabels \(\mu\leftrightarrow\nu\) and uses the
antisymmetry of \(E^{[\mu\nu]}\). Hence reality requires
\begin{equation*}
  E^{[\mu\nu]}
  =
  -
  \gamma^0\bigl(E^{[\mu\nu]}\bigr)^\dagger\gamma^0 .
\end{equation*}
Thus \(E^{[\mu\nu]}\) is Dirac-anti-Hermitian. We therefore write
\begin{equation}
  E^{[\mu\nu]}
  =
  \frac{i}{3}\,\mathcal E^{\mu\nu},
  \qquad
  \gamma^0\bigl(\mathcal E^{\mu\nu}\bigr)^\dagger\gamma^0
  =
  \mathcal E^{\mu\nu}.
  \label{eq:E-herm}
\end{equation}
The factor \(i\) converts a Dirac-Hermitian coefficient into a
Dirac-anti-Hermitian one. The factor \(1/3\) is a normalization convention
chosen for later formulae.

We now classify the possible Dirac-Hermitian, bivector-valued coefficients
\(\mathcal E^{\mu\nu}\). Expanding on the Clifford basis gives
\begin{equation}
  \mathcal E^{\mu\nu}
  =
  a^{\mu\nu}\mathbf 1
  +
  b^{\mu\nu}i\gamma^5
  +
  c^{\mu\nu}{}_{\alpha}\gamma^\alpha
  +
  d^{\mu\nu}{}_{\alpha}\gamma^\alpha\gamma^5
  +
  \frac12\,e^{\mu\nu}{}_{\alpha\beta}\sigma^{\alpha\beta}.
  \label{eq:E-clifford}
\end{equation}
The coefficient tensors are constant and may be built only from
\(\eta_{\mu\nu}\), \(\varepsilon_{\mu\nu\rho\sigma}\), and Kronecker deltas.

There is no nonzero Lorentz-invariant antisymmetric rank-two tensor. Hence
the scalar and pseudoscalar terms with \(a^{\mu\nu}\) and \(b^{\mu\nu}\) vanish.
Likewise, there is no invariant rank-three tensor with the index structure
needed for \(c^{\mu\nu}{}_\alpha\) or \(d^{\mu\nu}{}_\alpha\), so the vector
and axial-vector sectors vanish as well.

Thus only the tensor sector survives. Now, Lorentz covariance asks
for an invariant linear map from bivectors to bivectors. In four dimensions,
with orientation fixed, the only such maps are the identity and the Hodge
dual, therefore
\begin{equation*}
  e^{\mu\nu}{}_{\alpha\beta}
  =
  2\ell\,\delta^{[\mu}{}_{\alpha}\delta^{\nu]}{}_{\beta}
  +
  \ell_5\,\varepsilon^{\mu\nu}{}_{\alpha\beta},
\end{equation*}
and hence
\begin{equation}
  \mathcal E^{\mu\nu}
  =
  \ell\,\sigma^{\mu\nu}
  +
  \ell_5\, ^\star\sigma^{\mu\nu}.
\end{equation}
Both \(\sigma^{\mu\nu}\) and \( ^\star\sigma^{\mu\nu}\) are
Dirac-Hermitian, so \eqref{eq:E-herm} implies
\[
  \ell,\ell_5\in\mathbb R .
\]
The dimensions of \(\ell\) and \(\ell_5\) are those of length, equivalently
inverse mass.
\hfill\(\square\)

A useful current representative for the null deformation
\eqref{eq:DL-final} is the local current
\begin{equation*}
  F^\mu
  =
  \frac{i}{6}
  \left[
    \bar\psi\,\mathcal E^{\mu\nu}\partial_\nu\psi
    -
    (\partial_\nu\bar\psi)\,\mathcal E^{\mu\nu}\psi
  \right].
\end{equation*}
Using the constancy and antisymmetry of \(\mathcal E^{\mu\nu}\), one obtains
\begin{equation*}
  \partial_\mu F^\mu
  =
  \frac{i}{3}\,
  (\partial_\mu\bar\psi)\,
  \mathcal E^{\mu\nu}\,
  (\partial_\nu\psi)
  =
  \Delta\mathcal L .
\end{equation*}

\begin{proposition}[Regularizing null deformation]
\label{prop:regularizing-null}
The minimal first-order, charge-neutral, real, proper Poincar\'e-covariant
null deformations which regularize the covariant Dirac Legendre map are
\begin{equation}
  \mathcal L_{\mathrm{reg}}
  =
  \mathcal L_D
  +
  \frac{i}{3}\,
  (\partial_\mu\bar\psi)
  \left(
    \ell\,\sigma^{\mu\nu}
    +
    \ell_5\, ^\star\sigma^{\mu\nu}
  \right)
  (\partial_\nu\psi),
  \label{eq:L-reg-final}
\end{equation}
with
\begin{equation}
  \ell,\ell_5\in\mathbb R,
  \qquad
  (\ell,\ell_5)\neq(0,0).
  \label{eq:ell-condition-final}
\end{equation}
\end{proposition}

\noindent\textit{Proof.}
By Proposition~\ref{prop:lorentz-classification}, it only remains to impose regularity.

Since \(\mathcal L_D\) is linear in the field gradients, only the added null term contributes to the Hessian. Thus the only nonvanishing second derivatives of \(\mathcal L_{\mathrm{reg}}\) with respect to the field gradients are the mixed ones,
\begin{equation*}
\frac{\partial^2\mathcal L_{\mathrm{reg}}}
{\partial(\partial_\mu\bar\psi)\,\partial(\partial_\nu\psi)}
=
\frac{i}{3}\,\mathcal E^{\mu\nu},
\end{equation*}
while
\begin{equation*}
\frac{\partial^2\mathcal L_{\mathrm{reg}}}
{\partial(\partial_\mu\psi)\,\partial(\partial_\nu\psi)}
=
\frac{\partial^2\mathcal L_{\mathrm{reg}}}
{\partial(\partial_\mu\bar\psi)\,\partial(\partial_\nu\bar\psi)}
=
0.
\end{equation*}
Hence the Hessian is nondegenerate if and only if the mixed coefficient \(\mathcal E^{\mu\nu}\) is invertible.

To show this, use
\begin{equation}
^\star\sigma^{\mu\nu}=-i\,\sigma^{\mu\nu}\gamma^5,
\qquad
[\sigma^{\mu\nu},\gamma^5]=0,
\label{eq:dual5}
\end{equation}
so that
\begin{equation}
\mathcal E^{\mu\nu}
=
(\ell-i\ell_5\gamma^5)\sigma^{\mu\nu}.
\label{eq:E-factorized-proof}
\end{equation}
Now introduce \cite{struckmeier05}
\begin{equation}
\tau_{\mu\nu}
:=
\frac{i}{6}\Bigl(\gamma_\mu\gamma_\nu+3\gamma_\nu\gamma_\mu\Bigr),
\qquad
\tau_{\mu\alpha}\sigma^{\alpha\nu}
=
\delta_\mu{}^\nu\,\mathbf 1,
\label{eq:tau-def-proof}
\end{equation}
and define
\begin{equation}
U_{\mu\nu}
:=
\frac{1}{\ell^2+\ell_5^2}\,\tau_{\mu\nu}\,(\ell+i\ell_5\gamma^5).
\label{eq:U-def-proof}
\end{equation}
Using \eqref{eq:E-factorized-proof}, \eqref{eq:tau-def-proof}, and the commutativity of \(\gamma^5\) with \(\sigma^{\mu\nu}\), one finds
\begin{align}
U_{\mu\alpha}\mathcal E^{\alpha\nu}
&=
\frac{1}{\ell^2+\ell_5^2}\,
\tau_{\mu\alpha}(\ell+i\ell_5\gamma^5)(\ell-i\ell_5\gamma^5)\sigma^{\alpha\nu}
\nonumber\\
&=
\tau_{\mu\alpha}\sigma^{\alpha\nu}
=
\delta_\mu{}^\nu\,\mathbf 1.
\label{eq:left-inverse-proof}
\end{align}
Therefore \(U_{\mu\nu}\) is a left inverse of \(\mathcal E^{\mu\nu}\), and since \(\mathcal E^{\mu\nu}{}_A{}^B\) is a square \(16\times16\) matrix, it is invertible whenever
\begin{equation*}
\ell^2+\ell_5^2\neq 0.
\end{equation*}
Therefore regularity holds if and only if
\begin{equation*}
(\ell,\ell_5)\neq(0,0).
\end{equation*}
\hfill $\square$

Using
\begin{equation*}
\gamma_\alpha\sigma^{\alpha\beta}=3i\,\gamma^\beta
\quad \text{and} \quad
\sigma^{\alpha\beta}\gamma_\beta=3i\,\gamma^\alpha,
\end{equation*}
with $\lambda^2:=\ell^2+\ell_5^2$, the regularized Lagrangian \eqref{eq:L-reg-final} may be rewritten in the equivalent symmetric form
\begin{widetext}
\begin{equation}
\mathcal L_{\mathrm{reg}}
=
\frac{i}{3}
\left(
\partial_\mu\bar\psi
-\frac{i}{2\lambda^2}\,\bar\psi\,\gamma_\mu(\ell+i\ell_5\gamma^5)
\right)
\mathcal E^{\mu\nu}
\left(
\partial_\nu\psi
+\frac{i}{2\lambda^2}\,(\ell+i\ell_5\gamma^5)\gamma_\nu\psi
\right)
-\left(m-\frac{\ell}{\lambda^2}\right)\bar\psi\psi
-\frac{\ell_5}{\lambda^2}\,\bar\psi\,i\gamma^5\psi .
\label{eq:L-reg-symmetric}
\end{equation}
\end{widetext}

If parity (or CP--symmetry) is imposed in addition, the pseudotensor contribution is excluded and
\begin{equation*}
\ell_5=0 ,
\end{equation*}
one then recovers the parity-even regularization as the unique representative within our demands (up to an overall normalization) \cite{struckmeier2024pauli}.

\subsection{Corresponding covariant Hamiltonian}
\label{subsec:dirac-hamiltonian}

To show that the regularization procedure indeed yields a De Donder--Weyl Hamiltonian, consider the regularized Lagrangian \eqref{eq:L-reg-final}, with \(\mathcal E^{\mu\nu}\) as in Eq.~\eqref{eq:DL-final}. Treating \(\psi\) and \(\bar\psi\) as independent fields, the corresponding polymomenta are defined by
\begin{equation*}
\pi^\mu
:=
\frac{\partial\mathcal L_{\mathrm{reg}}}{\partial(\partial_\mu\bar\psi)},
\quad
\bar\pi^\mu
:=
\frac{\partial\mathcal L_{\mathrm{reg}}}{\partial(\partial_\mu\psi)}.
\end{equation*}
A direct calculation gives
\begin{equation}
\pi^\mu
=
-\frac{i}{2}\gamma^\mu\psi
+\frac{i}{3}\,\mathcal E^{\mu\nu}\partial_\nu\psi,
\quad
\bar\pi^\mu
=
\frac{i}{2}\bar\psi\gamma^\mu
+\frac{i}{3}\,(\partial_\nu\bar\psi)\,\mathcal E^{\nu\mu}.
\label{eq:pi-explicit}
\end{equation}

Using the inverse \eqref{eq:U-def-proof}, the field derivatives are recovered as
\begin{equation}
\partial_\mu\psi
=
\frac{3}{i}\,
U_{\mu\nu}
\left(
\pi^\nu+\frac{i}{2}\gamma^\nu\psi
\right),
\quad
\partial_\mu\bar\psi
=
\frac{3}{i}
\left(
\bar\pi^\nu-\frac{i}{2}\bar\psi\gamma^\nu
\right)
U_{\nu\mu}.
\label{eq:velocity-inversion}
\end{equation}

The De Donder--Weyl Hamiltonian density is
\begin{equation}
\mathcal H_{\mathrm{reg}}
=
\bar\pi^\mu\,\partial_\mu\psi
+
(\partial_\mu\bar\psi)\,\pi^\mu
-
\mathcal L_{\mathrm{reg}}.
\label{eq:H-def}
\end{equation}
Substituting \eqref{eq:pi-explicit} into \eqref{eq:H-def} and using \eqref{eq:velocity-inversion}, this becomes
\begin{equation}
\mathcal H_{\mathrm{reg}}
=
\frac{3}{i}
\left(
\bar\pi^\mu-\frac{i}{2}\bar\psi\gamma^\mu
\right)
U_{\mu\nu}
\left(
\pi^\nu+\frac{i}{2}\gamma^\nu\psi
\right)
+
m\,\bar\psi\psi ,
\label{eq:H-reg-expanded}
\end{equation}
thus the regularization procedure indeed yields a covariant Hamiltonian density. Setting $\ell_5=0$, (\ref{eq:H-reg-expanded}) reduces to the free Dirac covariant Hamiltonians of \cite{struckmeier2024extended,vasak2023covariant}. 

In the terminology of Lepagean regularization, the result may be viewed as a dedonderized
representative of the free Dirac variational class
\cite{krupkova2001legendre}: the added derivative-quadratic null term leaves
the free Euler--Lagrange equations unchanged, but changes the ordinary
Hamilton--De Donder Legendre map so that the derivative variables can be
recovered from the polymomenta. In the following sections we make the additional choice of taking
this regular representative itself as the matter Lagrangian before applying
minimal gauging, after which the null term is no longer variationally
trivial.

\section{\texorpdfstring{$U(1)$}{U(1)} Gauging and Emergent Dipole Moments}
\label{subsec:qed}

We now apply minimal $U(1)$ gauging to the Legendre-regularizing
representative
\begin{equation*}
  \Delta\mathcal L
  =
  \frac{i}{3}
  (\partial_\mu\bar\psi)\,
  \mathcal E^{\mu\nu}\,
  (\partial_\nu\psi),
  \qquad
  \mathcal E^{\mu\nu}
  =
  \ell\,\sigma^{\mu\nu}
  +
  \ell_5\, ^\star\sigma^{\mu\nu}.
\end{equation*}
The important point is that the regularity condition has selected a particular
derivative-quadratic representative of the free null class. Minimal coupling
is therefore applied to this representative.

We use the convention
\begin{equation*}
  \psi\mapsto e^{i\Lambda(x)}\psi,
  \quad
  \bar\psi\mapsto\bar\psi e^{-i\Lambda(x)},
  \quad
  A_\mu\mapsto A_\mu-\frac{1}{q}\partial_\mu\Lambda ,
\end{equation*}
with covariant derivatives
\begin{equation*}
  \nabla_\mu\psi
  :=
  \partial_\mu\psi+i q A_\mu\psi,
  \qquad
  \bar\nabla_\mu\bar\psi
  :=
  \partial_\mu\bar\psi-iq\bar\psi A_\mu .
\end{equation*}
Thus
\begin{equation*}
  [\nabla_\mu,\nabla_\nu]\psi
  =
  iqF_{\mu\nu}\psi,
  \qquad
  F_{\mu\nu}:=\partial_\mu A_\nu-\partial_\nu A_\mu .
\end{equation*}

The gauged regularized Dirac--Maxwell density is then
\begin{align}
  \mathcal L_{U(1)}
  &=
  \frac{i}{2}
  \left(
    \bar\psi\gamma^\mu\nabla_\mu\psi
    -
    (\bar\nabla_\mu\bar\psi)\gamma^\mu\psi
  \right)
  -
  m\bar\psi\psi
  -
  \frac14 F_{\mu\nu}F^{\mu\nu}
  \nonumber\\
  &\
  +
  \frac{i}{3}
  (\bar\nabla_\mu\bar\psi)\,
  \mathcal E^{\mu\nu}\,
  (\nabla_\nu\psi).
  \label{eq:L-reg-U1}
\end{align}
The spinorial Legendre sector remains regular, because the highest-derivative
part is still controlled by the same mixed coefficient
\(\mathcal E^{\mu\nu}\). Hence the condition remains
\begin{equation*}
  \ell^2+\ell_5^2\neq0 .
\end{equation*}

In the free theory, the last term in \eqref{eq:L-reg-U1} reduces to a total
divergence. After gauging, however, the covariant derivatives no longer
commute. Since \(\mathcal E^{\mu\nu}\) is antisymmetric and gauge neutral,
one finds
\begin{align}
  \frac{i}{3}
  (\bar\nabla_\mu\bar\psi)\,
  \mathcal E^{\mu\nu}\,
  (\nabla_\nu\psi)
  &=
  \frac{i}{3}\,
  \partial_\mu
  \left(
    \bar\psi\,\mathcal E^{\mu\nu}\nabla_\nu\psi
  \right)
  -
  \frac{i}{3}\,
  \bar\psi\,\mathcal E^{\mu\nu}\nabla_\mu\nabla_\nu\psi
  \nonumber\\
  &=
  \frac{i}{3}\,
  \partial_\mu
  \left(
    \bar\psi\,\mathcal E^{\mu\nu}\nabla_\nu\psi
  \right)
  +
  \frac{q}{6}\,
  F_{\mu\nu}\,
  \bar\psi\,\mathcal E^{\mu\nu}\psi ,
  \label{eq:gauged-null-descendant}
\end{align}
where the second term on the right-hand side now represents a genuine interaction. The Euler--Lagrange equations for the spinor fields are therefore
\begin{subequations}
\label{eq:Dirac-U1-general}
\begin{align}
  i\gamma^\mu\nabla_\mu\psi
  -
  m\psi
  +
  \frac{q}{6}
  F_{\mu\nu}
  \left(
    \ell\,\sigma^{\mu\nu}
    +
    \ell_5\, ^\star\sigma^{\mu\nu}
  \right)
  \psi
  &=0,
  \label{eq:Dirac-U1-general-psi}
  \\
  i(\bar\nabla_\mu\bar\psi)\gamma^\mu
  +
  m\bar\psi
  -
  \frac{q}{6}
  F_{\mu\nu}
  \bar\psi
  \left(
    \ell\,\sigma^{\mu\nu}
    +
    \ell_5\, ^\star\sigma^{\mu\nu}
  \right)
  &=0 .
  \label{eq:Dirac-U1-general-barpsi}
\end{align}
\end{subequations}
where the parameter \(\ell\) controls a Pauli-type magnetic dipole operator, whereas the parameter \(\ell_5\), a parity-odd electric-dipole-type operator.  In usual dipole notation, with the convention
\[
\mathcal L_{\rm dip}
=
\frac{q a_f}{4m_f}
F_{\mu\nu}\bar\psi\sigma^{\mu\nu}\psi
-\frac{i}{2}d_f
F_{\mu\nu}\bar\psi\sigma^{\mu\nu}\gamma^5\psi ,
\]
we may identify, up to the sign convention used for \(d_f\) and \(\gamma^5\), 
\[
a_f=\frac{2m_f}{3}\ell,
\qquad
d_f=\frac{q}{3}\ell_5.
\]

For a theory with several Dirac fields, the regularizing parameters should
be attached to the free-field representative of each species,
\[
        \ell\to\ell_f,\qquad
        \ell_5\to\ell_{5f}.
\]
These parameters are introduced before gauging: they define the
Legendre-regular representative of the free Dirac Lagrangian for the
field \(f\), not a coefficient associated only with electromagnetism.
Once that representative is chosen as the matter Lagrangian, subsequent
couplings of the same fermion inherit the same pair
\((\ell_f,\ell_{5f})\). Thus an electromagnetic dipole constraint does
not merely bound an electromagnetic parameter, it bounds the same
regularization lengths that also enter other (\textit{e.g.} weak, color, or geometric)
couplings of that fermion whenever those couplings are obtained from the
same gauged representative.

The electromagnetic descendant gives immediate bounds on these lengths.
Writing the electric charge as \(q_f=Q_f e\), the identifications above
give
\[
        |\ell_f|
        \lesssim
        {3\over 2m_f}\,\delta a_f,
        \qquad
        |\ell_{5f}|
        \lesssim
        {3\over |Q_f|}\,{|d_f|\over e}.
\]
Here \(\delta a_f\) denotes the allowed contribution to the anomalous
magnetic moment attributed to the regularizing Pauli term. A full
new-physics bound would require specifying the Standard-Model prediction
and input constants, but an experimental sensitivity estimate is obtained
directly from the measured uncertainty in \(a_f\).

For the parity-even parameter, precision magnetic-moment measurements give
a direct sensitivity estimate. Using the the experimental uncertainty
in the electron anomaly,
\[
        \delta a_e^{\rm exp}=1.3\times10^{-13},
\]
as the allowed size of an additional Pauli contribution, one obtains
\[
        |\ell_e|
        \lesssim
        {3\over 2m_e}\,
        1.3\times10^{-13}
        \simeq
        3.8\times10^{-10}\,{\rm GeV}^{-1}
\]
or equivalently
\[
        |\ell_e|^{-1}
        \gtrsim
        2.6\times10^{9}\,{\rm GeV}.
\]
The Standard-Model comparison of the electron magnetic moment is tied in
particular to the fine-structure constant input, so this number should be
read as an experimental sensitivity to \(\ell_e\), not as a global fit
bound \cite{Fan:2023ElectronMagneticMoment}.

The parity-odd parameter is constrained more directly by electric-dipole
searches. The current electron EDM bound,
\[
        |d_e|<4.1\times10^{-30}\,e\,{\rm cm}
\]
at \(90\%\) confidence level \cite{Roussy:2023eEDM}, gives, for
\(Q_e=-1\),
\[
        |\ell_{5e}|
        \lesssim
        6.2\times10^{-16}\,{\rm GeV}^{-1},
\]
or
\[
        |\ell_{5e}|^{-1}
        \gtrsim
        1.6\times10^{15}\,{\rm GeV},
\]
assuming no cancellation with other CP-odd contributions.

Using
\eqref{eq:dual5}
the gauge-induced dipole term can also be written as
\begin{equation*}
  \frac{q}{6}
  F_{\mu\nu}\bar\psi\mathcal E^{\mu\nu}\psi
  =
  \frac{q}{6}
  F_{\mu\nu}\,
  \bar\psi\sigma^{\mu\nu}
  \left(
    \ell-i\ell_5\gamma^5
  \right)
  \psi .
\end{equation*}

Equivalently, the \(\ell_5\) contribution may be absorbed into the Hodge dual
of the electromagnetic field strength. Defining
\begin{equation*}
  ({}^\star F)_{\mu\nu}
  :=
  \frac12\varepsilon_{\mu\nu\rho\sigma}F^{\rho\sigma},
\end{equation*}
we see that
\begin{equation}
  \frac{q}{6}
  F_{\mu\nu}\bar\psi\mathcal E^{\mu\nu}\psi
  =
  \frac{q}{6}
  \bar\psi\sigma^{\mu\nu}
  \left(
    \ell F_{\mu\nu}
    +
    \ell_5({}^\star F)_{\mu\nu}
  \right)
  \psi .
  \label{eq:pauli-dual-F-form}
\end{equation}

Varying with respect to the gauge potential gives the Maxwell equation
\begin{equation}
  \partial_\mu F^{\mu\nu}
  =
  q\bar\psi\gamma^\nu\psi
  +
  \frac{q}{3}
  \partial_\mu
  \left[
    \bar\psi
    \left(
      \ell\,\sigma^{\mu\nu}
      +
      \ell_5\, ^\star\sigma^{\mu\nu}
    \right)
    \psi
  \right].
  \label{eq:maxwell-general-regularized}
\end{equation}
It is useful to decompose the electromagnetic source into the minimal Dirac
current
\begin{equation*}
J_{\rm cov}^\nu
:=
q\bar\psi\gamma^\nu\psi
\end{equation*}
and the magnetic- and electric-dipole-type currents
\begin{equation*}
J_{\rm M}^\nu
:=
\partial_\mu\mathcal M_{\rm M}^{\mu\nu},
\qquad
J_{\rm E}^\nu
:=
\partial_\mu\mathcal M_{\rm E}^{\mu\nu},
\end{equation*}
where
\begin{equation*}
\mathcal M_{\rm M}^{\mu\nu}
:=
\frac{q\ell}{3}\bar\psi\sigma^{\mu\nu}\psi,
\qquad
\mathcal M_{\rm E}^{\mu\nu}
:=
\frac{q\ell_5}{3}\bar\psi\,{}^\star\sigma^{\mu\nu}\psi .
\end{equation*}
Thus
\begin{equation}
J_{\rm dip}^\nu
:=
J_{\rm M}^\nu+J_{\rm E}^\nu
=
\partial_\mu\mathcal M^{\mu\nu},
\quad
\mathcal M^{\mu\nu}
:=
\mathcal M_{\rm M}^{\mu\nu}
+
\mathcal M_{\rm E}^{\mu\nu},
\label{eq:dipole-current}
\end{equation}
and \eqref{eq:maxwell-general-regularized} becomes
\begin{equation}
\partial_\mu F^{\mu\nu}
=
J_{\rm tot}^\nu,
\qquad
J_{\rm tot}^\nu
:=
J_{\rm cov}^\nu+J_{\rm dip}^\nu .
\label{eq:maxwell-total-current}
\end{equation}

The two dipole currents are identically conserved. Indeed, the antisymmetry
of the superpotentials \(\mathcal M_{\rm M}^{\mu\nu}\) and
\(\mathcal M_{\rm E}^{\mu\nu}\) gives
\begin{equation*}
\partial_\nu J_{\rm M}^\nu
=
\partial_\nu\partial_\mu\mathcal M_{\rm M}^{\mu\nu}
\equiv0,
\quad
\partial_\nu J_{\rm E}^\nu
=
\partial_\nu\partial_\mu\mathcal M_{\rm E}^{\mu\nu}
\equiv0 ,
\end{equation*}
therefore,
\begin{equation*}
\partial_\nu J_{\rm dip}^\nu
=
\partial_\nu J_{\rm M}^\nu
+
\partial_\nu J_{\rm E}^\nu
=
\partial_\nu\partial_\mu\mathcal M^{\mu\nu}
\equiv0 .
\end{equation*}
These are off-shell identities, unlike the minimal Dirac current, which is
conserved only on shell, as one obtains
\begin{equation*}
\partial_\nu J_{\rm cov}^\nu=0,
\end{equation*}
upon use of the spinor equations of motion.

The superpotential form also implies
\begin{equation*}
Q_{\rm M,E}
:=
\int_\Sigma d^3x\,J_{\rm M,E}^0
=
\int_{\partial\Sigma}dS_i\,
\mathcal M_{\rm M,E}^{i0}.
\end{equation*}
The total dipole charge is therefore
\begin{align*}
Q_{\rm dip}
&:=
Q_{\rm M}+Q_{\rm E}
=
\int_\Sigma d^3x\,J_{\rm dip}^0
\nonumber\\
&=
\int_{\partial\Sigma}dS_i\,
\mathcal M^{i0}.
\end{align*}
For localized fields these boundary terms vanish, so that
\begin{equation*}
Q_{\rm M}=Q_{\rm E}=Q_{\rm dip}=0.
\end{equation*}
The dipole currents therefore modify the local polarization--magnetization
structure of the source without changing its total \(U(1)\) charge.
This superpotential structure is the Abelian analogue of the
Gordon polarization-current decomposition of the Dirac--Yang--Mills
Noether current, where the polarization current is the exterior
covariant derivative of a moment two-form
\cite{Gordon:1928,Hehl:1997currents}. In the present construction,
however, the dipole current is obtained by gauging the
Legendre-regular null representative, rather than by an on-shell
Gordon rewrite of the minimally coupled Dirac current.

If parity or CP invariance is imposed, then \(\ell_5=0\), and
\eqref{eq:pauli-dual-F-form} reduces to the anomalous Pauli coupling term
found in \cite{struckmeier2024pauli}.

\section{Gauging to Metric-Affine Spacetime}
\label{sec:spacetime-metricaffine}

We now replace the fixed Minkowski background by a dynamical
metric-affine spacetime geometry. In this section, Greek indices refer to
curved spacetime components, while the flat indices of the preceding
Minkowski discussion are carried by local Lorentz-frame Latin indices.
The metric tensor is written in terms of tetrads as
\begin{equation*}
g_{\mu\nu}(x)
=
\eta_{ab}\,e^a{}_\mu(x)e^b{}_\nu(x),
\quad
e:=\det(e^a{}_\mu)=\sqrt{-g}.
\end{equation*}
The curved Dirac matrices are
\begin{equation*}
\gamma^\mu(x):=e_a{}^\mu(x)\gamma^a,
\qquad
\sigma^{\mu\nu}(x)
:=
\frac{i}{2}\big[\gamma^\mu(x),\gamma^\nu(x)\big].
\end{equation*}
We denote the corresponding Hodge-dual Clifford bivector by
\begin{equation}
{}^\star\sigma^{\mu\nu}(x)
:=
\frac12\,\varepsilon^{\mu\nu\rho\sigma}(x)\sigma_{\rho\sigma}(x)
=
-i\sigma^{\mu\nu}(x)\gamma^5 ,
\label{eq:dual-sigma-star}
\end{equation}
with the sign fixed by the conventions of Sec.~\ref{sec:regular-dirac}.
The two-parameter regularizing bivector is
\begin{equation}
\mathcal E^{\mu\nu}(x)
:=
\ell\,\sigma^{\mu\nu}(x)
+
\ell_5\,{}^\star\sigma^{\mu\nu}(x).
\label{eq:E-metricaffine}
\end{equation}

The affine connection is denoted by \(\Gamma^\rho{}_{\mu\nu}\), with no
symmetry or metric-compatibility condition imposed. We use
\begin{equation*}
S^\lambda{}_{\mu\nu}
:=
\Gamma^\lambda{}_{[\mu\nu]},
\qquad
Q_{\lambda\mu\nu}
:=
-\nabla_\lambda g_{\mu\nu}
=
-g_{\mu\nu;\lambda}
\end{equation*}
for the Cartan torsion and nonmetricity tensors respectively.

The tetrad and affine connections are related by the tetrad postulate
\begin{equation}
\partial_\mu e^a{}_{\nu}
-
\Gamma^\rho{}_{\nu\mu}e^a{}_\rho
+
\omega_\mu{}^a{}_b e^b{}_\nu
=
0 .
\label{eq:tetrad-postulate-RC}
\end{equation}
At this stage no symmetry is imposed on the frame indices of
\(\omega_\mu{}^a{}_b\).

Spinor derivatives are defined by
\begin{equation*}
D_\mu\psi
:=
\partial_\mu\psi+\Omega_\mu\psi,
\qquad
\bar D_\mu\bar\psi
:=
\partial_\mu\bar\psi-\bar\psi\Omega_\mu .
\end{equation*}
Where we use the standard Lorentz-spinor realization,
\begin{equation}
\Omega_\mu
=
-\frac{i}{4}\sigma^{ab}\omega_{ab\mu},
\label{eq:spinor-connection-metricaffine}
\end{equation}
in which nonmetricity enters through
the tetrad, the density factors, and the covariant differentiation of
Clifford-valued tensors, while the spinor connection itself contains only
the Lorentz part \(\omega_{[ab]\mu}\). General metric-affine gauge theory, including the
role of hypermomentum and the world-spinor problem, is reviewed in
\cite{Hehl:1995ue}.

Minimal spacetime gauging of the two-parameter Legendre-regular Dirac
representative gives
\begin{align}
\tilde{\mathcal L}_{\text{MA}}
&=
e\Bigg[
\frac{i}{2}\bar\psi\gamma^\mu D_\mu\psi
-
\frac{i}{2}(\bar D_\mu\bar\psi)\gamma^\mu\psi
-
m\bar\psi\psi\nonumber\\
&\quad+
\frac{i}{3}
(\bar D_\mu\bar\psi)
\mathcal E^{\mu\nu}
D_\nu\psi
\Bigg]
+
\tilde{\mathcal L}_{\rm Gr}.
\label{eq:LD-E-metricaffine}
\end{align}
The gravitational density \(\tilde{\mathcal L}_{\rm Gr}\) is left unspecified.
Spinorial sector regularity is still controlled by the same mixed Hessian
coefficient \(\mathcal E^{\mu\nu}\), hence
\begin{equation*}
\ell^2+\ell_5^2\neq0 .
\end{equation*}

Variation of \eqref{eq:LD-E-metricaffine} with respect to \(\bar\psi\) first
gives the non-manifestly covariant equation
\begin{align}
&
\left[
i\gamma^\beta
-
\frac{i}{3}
\left(
\partial_\alpha\mathcal E^{\alpha\beta}
+
\Omega_\alpha\mathcal E^{\alpha\beta}
-
\mathcal E^{\alpha\beta}\Omega_\alpha
+
\mathcal E^{\alpha\beta}\frac1e\partial_\alpha e
\right)
\right]D_\beta\psi
\nonumber\\
&
-
\Bigg[
m
-
\frac{i}{2}
\left(
\partial_\alpha\gamma^\alpha
+
\Omega_\alpha\gamma^\alpha
-
\gamma^\alpha\Omega_\alpha
+
\gamma^\alpha\frac1e\partial_\alpha e
\right)
\nonumber\\
&\quad
+
\frac{i}{3}
\mathcal E^{\alpha\beta}
\left(
\partial_\alpha\Omega_\beta
+
\Omega_\alpha\Omega_\beta
\right)
\Bigg]\psi
=0 .
\label{eq:metric-affine-nonmanifest-E}
\end{align}
The spin-connection-curvature is
\begin{align*}
&\mathfrak R_{\alpha\beta}
:=
\partial_\alpha\Omega_\beta
-
\partial_\beta\Omega_\alpha
+
\Omega_\alpha\Omega_\beta
-
\Omega_\beta\Omega_\alpha,\\
&[D_\alpha,D_\beta]\psi=\mathfrak R_{\alpha\beta}\psi .
\end{align*}
Since \(\mathcal E^{\alpha\beta}=-\mathcal E^{\beta\alpha}\),
\begin{equation*}
\mathcal E^{\alpha\beta}
\left(
\partial_\alpha\Omega_\beta
+
\Omega_\alpha\Omega_\beta
\right)
=
\frac12\mathcal E^{\alpha\beta}\mathfrak R_{\alpha\beta}.
\end{equation*}

Following \cite{struckmeier2024pauli}, tildes denote relative tensor
densities of weight one,
\begin{equation*}
\tilde\gamma^\mu:=e\gamma^\mu,
\qquad
\tilde{\mathcal E}^{\mu\nu}:=e\mathcal E^{\mu\nu},
\end{equation*}
with
\begin{equation*}
e_{;\alpha}
=
\partial_\alpha e
-
e\,\Gamma^\rho{}_{\rho\alpha}
\end{equation*}
and
\begin{equation*}
\gamma^\mu{}_{;\alpha}
=
\partial_\alpha\gamma^\mu
+\left[\Omega_\alpha ,\gamma^\mu\right]
+\Gamma^\mu{}_{\lambda\alpha}\gamma^\lambda .
\end{equation*}
The bracket terms in \eqref{eq:metric-affine-nonmanifest-E} may then be
written as
\begin{align*}
&\partial_\alpha\gamma^\alpha
+\Omega_\alpha\gamma^\alpha
-\gamma^\alpha\Omega_\alpha
+\gamma^\alpha\frac1e\partial_\alpha e
\nonumber\\
&\hspace{2em}
=
\frac1e\,\tilde\gamma^\alpha{}_{;\alpha}
-
2\gamma^\alpha S^\beta{}_{\alpha\beta},
\\[0.5em]
&\partial_\alpha\mathcal E^{\alpha\beta}
+\Omega_\alpha\mathcal E^{\alpha\beta}
-\mathcal E^{\alpha\beta}\Omega_\alpha
+\mathcal E^{\alpha\beta}\frac1e\partial_\alpha e
\nonumber\\
&\hspace{2em}
=
\frac1e\,\tilde{\mathcal E}^{\alpha\beta}{}_{;\alpha}
-
2\mathcal E^{\xi\beta}S^\alpha{}_{\xi\alpha}
+
\mathcal E^{\alpha\xi}S^\beta{}_{\alpha\xi}.
\end{align*}
Thus \eqref{eq:metric-affine-nonmanifest-E} becomes
\begin{align}
&
\left[
i\gamma^\beta
-
\frac{i}{3}
\left(
\frac1e\,\tilde{\mathcal E}^{\alpha\beta}{}_{;\alpha}
-
2\mathcal E^{\xi\beta}S^\alpha{}_{\xi\alpha}
+
\mathcal E^{\alpha\xi}S^\beta{}_{\alpha\xi}
\right)
\right]D_\beta\psi
\nonumber\\
&
-
\left[
m
-
\frac{i}{2}
\left(
\frac1e\,\tilde{\gamma}^{\alpha}{}_{;\alpha}
-
2\gamma^\alpha S^\beta{}_{\alpha\beta}
\right)
+
\frac{i}{6}\mathcal E^{\alpha\beta}\mathfrak R_{\alpha\beta}
\right]\psi
=0 .
\label{eq:metric-affine-master-E}
\end{align}
Equivalently,
\begin{equation}
\big(i\Gamma^\beta_{\rm MA}D_\beta-M_{\rm MA}\big)\psi=0,
\label{eq:GenDirac-metricaffine}
\end{equation}
where
\begin{equation}
\Gamma^\beta_{\rm MA}
:=
\gamma^\beta
-
\frac{1}{3}
\left(
\frac{1}{e}\,
\tilde{\mathcal E}^{\alpha\beta}{}_{;\alpha}
-
2\mathcal E^{\xi\beta}S^\alpha{}_{\xi\alpha}
+
\mathcal E^{\alpha\xi}S^\beta{}_{\alpha\xi}
\right),
\label{eq:Gamma-metricaffine}
\end{equation}
and
\begin{equation}
M_{\rm MA}
:=
m
-
\frac{i}{2}
\left(
\frac{1}{e}\,
\tilde\gamma^\alpha{}_{;\alpha}
-
2\gamma^\alpha S^\beta{}_{\alpha\beta}
\right)
+
\frac{i}{6}\mathcal E^{\alpha\beta}\mathfrak R_{\alpha\beta}.
\label{eq:M-metricaffine}
\end{equation}

\subsection{Metric-compatible torsionful limit}
\label{subsec:metric-compatible-torsionful}

We first specialize the metric-affine equation
\eqref{eq:metric-affine-master-E} to a metric-compatible geometry with
torsion. Thus
\begin{equation*}
Q_{\lambda\mu\nu}=0,
\qquad
S^\lambda{}_{\mu\nu}\neq0 .
\end{equation*}
Together with the tetrad postulate \eqref{eq:tetrad-postulate-RC}, metric
compatibility implies \(\omega_{(ab)\mu}=0\), so the full tetrad connection
reduces to its Lorentz part. These conditions imply covariant constancy of
the curved Clifford generators and of the regularizing bivector,
\begin{equation*}
\gamma^\mu{}_{;\lambda}=0,
\qquad
\sigma^{\mu\nu}{}_{;\lambda}=0,
\qquad
{}^\star\sigma^{\mu\nu}{}_{;\lambda}=0,
\qquad
\mathcal E^{\mu\nu}{}_{;\lambda}=0 .
\end{equation*}
With the density convention of Sec.~\ref{sec:spacetime-metricaffine}, the
density-divergence terms in \eqref{eq:metric-affine-master-E} then drop out,
and the Dirac equation becomes
\begin{align}
&
\Bigg[
i\gamma^\beta
+
\frac{i}{3}
\left(
2\mathcal E^{\xi\beta}S^\alpha{}_{\xi\alpha}
-
\mathcal E^{\xi\alpha}S^\beta{}_{\xi\alpha}
\right)
\Bigg]D_\beta\psi
\nonumber\\
&
-
\Bigg[
m
+
i\gamma^\beta S^\alpha{}_{\beta\alpha}
+
\frac{i}{6}\mathcal E^{\alpha\beta}\mathfrak R_{\alpha\beta}
\Bigg]\psi
=0 .
\label{eq:dirac-RC-E-spin-curvature}
\end{align}
In the Riemann--Cartan case the spin curvature is related to the
Riemann--Cartan curvature tensor by
\begin{equation}
\mathfrak R_{\alpha\beta}
=
-\frac{i}{4}\sigma^{\mu\nu}R_{\mu\nu\alpha\beta}.
\label{eq:spin-curvature-RC-relation}
\end{equation}
Hence \eqref{eq:dirac-RC-E-spin-curvature} may be written as
\begin{align}
&
\Bigg[
i\gamma^\beta
+
\frac{i}{3}
\left(
2\mathcal E^{\xi\beta}S^\alpha{}_{\xi\alpha}
-
\mathcal E^{\xi\alpha}S^\beta{}_{\xi\alpha}
\right)
\Bigg]D_\beta\psi
\nonumber\\
&
-
\Bigg[
m
+
i\gamma^\beta S^\alpha{}_{\beta\alpha}
+
\frac{1}{24}\mathcal E^{\alpha\beta}\sigma^{\mu\nu}
R_{\mu\nu\alpha\beta}
\Bigg]\psi
=0 ,
\label{eq:dirac-RC-E}
\end{align}
where setting \(\ell_5=0\), so that
\(\mathcal E^{\mu\nu}=\ell\sigma^{\mu\nu}\), gives the
metric-compatible equation of \cite{struckmeier2024pauli} with the corrected
torsion-divergence sign. 

From \eqref{eq:dirac-RC-E}, it is immediate that the formal limit
$(\ell,\ell_5)\to(0,0)$ yields
\begin{equation}
\Big(
i\gamma^{\beta}D_{\beta}
-m
-i\gamma^{\beta}S^{\alpha}{}_{\beta\alpha}
\Big)\psi
=0,
\label{eq:semi-Hehl-Datta}
\end{equation}
where $S^{\alpha}{}_{\beta\alpha}$ is the torsion trace vector.
To make the irreducible torsion content of this equation explicit,
\eqref{eq:tetrad-postulate-RC} implies
\begin{equation}
\label{eq:omega-from-tetrad-postulate}
\omega_{\nu}{}^{a}{}_{b}
=
e_{b}{}^{\mu}
\left(
\Gamma^{\rho}{}_{\mu\nu}e^{a}{}_{\rho}
-\partial_{\nu}e^{a}{}_{\mu}
\right).
\end{equation}
Inserting the Riemann--Cartan decomposition
\begin{equation*}
\Gamma^{\rho}{}_{\mu\nu}
=
\overset{\circ}{\Gamma}{}^{\rho}{}_{\mu\nu}
+
K^{\rho}{}_{\mu\nu}
\end{equation*}
into \eqref{eq:omega-from-tetrad-postulate} yields
\begin{equation*}
\omega_{\nu}{}^{a}{}_{b}
=
\underbrace{
e_{b}{}^{\mu}
\left(
\overset{\circ}{\Gamma}{}^{\rho}{}_{\mu\nu}e^{a}{}_{\rho}
-\partial_{\nu}e^{a}{}_{\mu}
\right)
}_{\overset{\circ}{\omega}{}_{\nu}{}^{a}{}_{b}}
+
\underbrace{
e_{b}{}^{\mu}
K^{\rho}{}_{\mu\nu}
e^{a}{}_{\rho}
}_{K_{\nu}{}^{a}{}_{b}},
\end{equation*}
and hence
\begin{equation*}
\omega_{\mu}{}^{ab}
=
\overset{\circ}{\omega}{}_{\mu}{}^{ab}
+
K_{\mu}{}^{ab}.
\end{equation*}

The spinor connection \eqref{eq:spinor-connection-metricaffine} thus decomposes as
\begin{equation*}
\Omega_{\mu}
=
\underbrace{
\left(
-\frac{i}{4}
\sigma^{ab}\overset{\circ}{\omega}{}_{\mu ab}
\right)
}_{\overset{\circ}{\Omega}{}_{\mu}}
-
\frac{i}{4}\sigma^{ab}K_{\mu ab}.
\end{equation*}
Therefore, with
\begin{equation*}
\overset{\circ}{D}_{\mu}\psi
:=
\partial_{\mu}\psi
+
\overset{\circ}{\Omega}{}_{\mu}\psi,
\end{equation*}
one obtains
\begin{equation*}
D_{\mu}\psi
=
\overset{\circ}{D}_{\mu}\psi
-
\frac{i}{4}\sigma^{ab}K_{\mu ab}\psi.
\end{equation*}

Contracting with $\gamma^{\mu}$ and using the Clifford algebra, the
contortion contribution separates into a torsion-trace term and a
purely axial term \cite{delhom2020minimal}:
\begin{equation}
\label{eq:gammaD-decomp}
\gamma^{\mu}D_{\mu}\psi
=
\gamma^{\mu}\overset{\circ}{D}_{\mu}\psi
+
\gamma^{\mu}S^{\alpha}{}_{\mu\alpha}\psi
-
T\psi,
\end{equation}
where $T$ depends only on the totally antisymmetric component
$S_{[\mu\nu\rho]}$ of the torsion tensor,
\begin{equation*}
T
:=
\frac{i}{4}
\varepsilon_{\mu\nu\rho\sigma}
S^{\mu\nu\rho}
\gamma^{\sigma}\gamma^{5}.
\end{equation*}
Substituting \eqref{eq:gammaD-decomp} into
\eqref{eq:semi-Hehl-Datta}, the torsion-trace terms cancel, giving
\begin{equation}
\label{eq:EC-Dirac-axial}
\Big(
i\gamma^{\mu}\overset{\circ}{D}_{\mu}
-iT
-m
\Big)\psi
=0.
\end{equation}
This is the standard Einstein--Cartan Dirac equation with axial
torsion coupling. Upon eliminating torsion by means of the
Einstein--Cartan connection equation, it becomes the usual nonlinear
Hehl--Datta equation \cite{hehl1971nonlinear}. Thus, in the formal
limit $(\ell,\ell_5)\to(0,0)$, the minimally coupled Dirac equation is
insensitive to the non-axial irreducible components of torsion. The
additional non-axial torsion couplings in
\eqref{eq:dirac-RC-E} are induced terms proportional to $\ell$ or
$\ell_5$.

For purely axial, or totally antisymmetric, torsion,
$S^{\alpha}{}_{\beta\alpha}=0$, and
\eqref{eq:dirac-RC-E} reduces to
\begin{align}
0
&=
\underbrace{
\left(
i\gamma^\beta D_\beta-m
\right)\psi
}_{\text{Einstein--Cartan Dirac term}}
\nonumber\\
&\quad
-
\underbrace{
\left[
\frac{i}{3}
\mathcal E^{\xi\alpha}
S^\beta{}_{\xi\alpha}D_\beta
+
\frac{1}{24}
\mathcal E^{\alpha\beta}
\sigma^{\mu\nu}R_{\mu\nu\alpha\beta}
\right]\psi
}_{\text{induced dipole-type corrections}} .
\label{eq:HD-Q-E-general}
\end{align}

Remarkably, for $\ell_5=0$, \eqref{eq:HD-Q-E-general} reproduces, up to the identification of coupling parameters, the Pauli-type spin-curvature and spin-torsion structures described in Obukhov \textit{et al.} in the (axial torsion) Poincar\'e gauge framework when gravitational dipole moments, obtained from the Gordon decomposition of the Noether currents \cite{Hehl:1995ue}, are introduced as explicit {nonminimal} matter couplings \cite{obukhov2014spin}.
Here, by contrast, the same operator content arises from {minimal} coupling applied to the nondegenerate
Dirac Lagrangian \eqref{eq:L-reg-final}. 
The explicit dictionary and the term-by-term matching are collected in Appendix~\ref{app-obkhv}.

It is convenient to rewrite \eqref{eq:dirac-RC-E} in more explicit form. Using the Clifford-algebra identity of Appendix \ref{app:sigmaR}
\begin{equation*}
\sigma^{\alpha\beta}\sigma^{\mu\nu}R_{\mu\nu\alpha\beta}
=2R+2iP\,\gamma^5+4i\sigma^{\mu\nu}R_{[\mu\nu]},
\end{equation*}
where the right-hand side contains, respectively, scalar, pseudoscalar and tensorial terms. Here, we know that \cite{struckmeier21a}
\begin{equation}
R_{[\mu\nu]} =
S^{\alpha}{}_{\mu\nu;\alpha}
- 2\, S^{\beta}{}_{\mu\nu}\, S^{\alpha}{}_{\beta\alpha}
+ 2\, S^{\alpha}{}_{[\mu|\alpha;|\nu]}\,.
\label{R-antisym}
\end{equation}
The pseudoscalar $P$ is proportional to the standard Riemann-Cartan curvature pseudoscalar $X$ often used in the EC and PGT literature \cite{obukhov1997chiral}, with our conventions
\begin{equation*}
P=12\,X,\qquad 
X:=\frac{1}{4!}\,\varepsilon^{\mu\nu\rho\sigma}R_{\mu\nu\rho\sigma}.
\end{equation*}
Finally, we may write \eqref{eq:dirac-RC-E} in compact form
\begin{equation}
\left(
i\Gamma^\beta_{\rm RC}D_\beta
-
M_{\rm RC}
\right)\psi=0,
\label{eq:GenDirac-RC-E}
\end{equation}
where
\begin{equation}
\Gamma^\beta_{\rm RC}
:=
\gamma^\beta
+
\frac{1}{3}
\left(
2\mathcal E^{\xi\beta}S^\alpha{}_{\xi\alpha}
-
\mathcal E^{\xi\alpha}S^\beta{}_{\xi\alpha}
\right),
\label{eq:Gamma-RC-E}
\end{equation}
and
\begin{align}
M_{\rm RC}
&=
m
+
i\gamma^\beta S^\alpha{}_{\beta\alpha}
\nonumber\\
&\quad
+
\frac{\ell}{12}
\left(
R
+
iP\gamma^5
+
2i\sigma^{\mu\nu}R_{[\mu\nu]}
\right)
\nonumber\\
&\quad
+
\frac{\ell_5}{12}
\left(
P
-
iR\gamma^5
+
2\sigma^{\mu\nu}R_{[\mu\nu]}\gamma^5
\right).
\label{eq:M-RC-expanded-E}
\end{align}

\subsection{Torsionless nonmetric limit}
\label{subsec:torsionless-nonmetric}

We now consider the complementary limit in which torsion vanishes but metric
compatibility is not imposed. Thus
\begin{equation*}
S^\lambda{}_{\mu\nu}=0,
\qquad
Q_{\lambda\mu\nu}\neq0 .
\end{equation*}
Equation~\eqref{eq:metric-affine-master-E} then reduces to
\begin{align}
&
\left[
i\gamma^\beta
-
\frac{i}{3e}\,
\tilde{\mathcal E}^{\alpha\beta}{}_{;\alpha}
\right]D_\beta\psi
\nonumber\\
&
-
\left[
m
-
\frac{i}{2e}\,
\tilde\gamma^\alpha{}_{;\alpha}
+
\frac{i}{6}\mathcal E^{\alpha\beta}\mathfrak R_{\alpha\beta}
\right]\psi
=
0 .
\label{eq:dirac-N-E-density}
\end{align}

To make the nonmetricity dependence explicit, we introduce the two traces
\begin{equation}
Q_\lambda:=Q_{\lambda\mu}{}^\mu,
\qquad
\widehat Q_\lambda:=Q^\mu{}_{\mu\lambda}.
\label{eq:Q-traces-N}
\end{equation}
With the convention \(Q_{\lambda\mu\nu}=-\nabla_\lambda g_{\mu\nu}\), one has
\begin{equation*}
\gamma^\beta{}_{;\alpha}
=
\frac12 Q_\alpha{}^\beta{}_\lambda\gamma^\lambda .
\end{equation*}
Since \(\gamma^5{}_{;\alpha}=0\) in the Lorentz-spinor realization, the same
relation holds for the dual spin bivector, and therefore
\begin{equation*}
\mathcal E^{\mu\nu}{}_{;\alpha}
=
\frac12 Q_\alpha{}^\mu{}_\lambda\mathcal E^{\lambda\nu}
+
\frac12 Q_\alpha{}^\nu{}_\lambda\mathcal E^{\mu\lambda}.
\end{equation*}
It follows that
\begin{equation}
\frac1e\tilde\gamma^\alpha{}_{;\alpha}
=
\frac12
(\widehat Q_\lambda-Q_\lambda)\gamma^\lambda ,
\label{eq:gamma-density-nonmetric}
\end{equation}
and
\begin{align}
\frac1e\tilde{\mathcal E}^{\alpha\beta}{}_{;\alpha}
&=
\frac12
\left[
(\widehat Q_\lambda-Q_\lambda)\mathcal E^{\lambda\beta}
+
Q_\alpha{}^\beta{}_\lambda\mathcal E^{\alpha\lambda}
\right].
\label{eq:E-density-nonmetric}
\end{align}
Substituting \eqref{eq:gamma-density-nonmetric} and
\eqref{eq:E-density-nonmetric} into \eqref{eq:dirac-N-E-density} gives
\begin{align}
&
\Bigg[
i\gamma^\beta
-
\frac{i}{6}
\left(
(\widehat Q_\lambda-Q_\lambda)\mathcal E^{\lambda\beta}
+
Q_\alpha{}^\beta{}_\lambda\mathcal E^{\alpha\lambda}
\right)
\Bigg]D_\beta\psi
\nonumber\\
&
-
\Bigg[
m
-
\frac{i}{4}
(\widehat Q_\lambda-Q_\lambda)\gamma^\lambda
+
\frac{i}{6}\mathcal E^{\alpha\beta}\mathfrak R_{\alpha\beta}
\Bigg]\psi
=
0 .
\label{eq:dirac-N-E-explicit}
\end{align}

Equivalently, the torsionless nonmetric equation may be written as
\begin{equation}
\left(
i\Gamma^\beta_{\rm N}D_\beta
-
M_{\rm N}
\right)\psi
=
0,
\label{eq:GenDirac-N-E}
\end{equation}
where
\begin{equation}
\Gamma^\beta_{\rm N}
:=
\gamma^\beta
-
\frac16
\left[
(\widehat Q_\lambda-Q_\lambda)\mathcal E^{\lambda\beta}
+
Q_\alpha{}^\beta{}_\lambda\mathcal E^{\alpha\lambda}
\right],
\label{eq:Gamma-N-E}
\end{equation}
and
\begin{equation}
M_{\rm N}
:=
m
-
\frac{i}{4}
(\widehat Q_\lambda-Q_\lambda)\gamma^\lambda
+
\frac{i}{6}\mathcal E^{\alpha\beta}\mathfrak R_{\alpha\beta}.
\label{eq:M-N-E}
\end{equation}

For the formal limit \((\ell,\ell_5)\to(0,0)\), \eqref{eq:dirac-N-E-explicit}
reduces to the standard torsionless nonmetric Dirac equation in the
Lorentz-spinor prescription,
\begin{equation}
\left[
i\gamma^\beta D_\beta
-
m
+
\frac{i}{4}
(\widehat Q_\lambda-Q_\lambda)\gamma^\lambda
\right]\psi=0 .
\label{eq:ordinary-nonmetric-dirac-limit}
\end{equation}
which reproduces, for example, the spinor equation obtained in tetrad-affine
\(f(\mathcal Q)\) gravity in \cite{vignolo2022spinor}.

We leave the spin-curvature term
\(\mathcal E^{\alpha\beta}\mathfrak R_{\alpha\beta}\) unexpanded in this
limit, since the Riemann--Cartan curvature identities used in
Sec.~\ref{subsec:metric-compatible-torsionful} rely on metric compatibility
and do not directly apply to a general nonmetric connection.

\subsection{Levi--Civita limit}
\label{subsec:levi-civita-limit}

The usual torsionless and metric-compatible limit is
\begin{equation*}
S^\lambda{}_{\mu\nu}=0,
\qquad
Q_{\lambda\mu\nu}=0.
\end{equation*}
The affine connection and spinor derivative then reduce to their
Levi--Civita counterparts,
\begin{equation*}
\Gamma^\rho{}_{\mu\nu}
=
\overset{\circ}{\Gamma}{}^\rho{}_{\mu\nu},
\qquad
D_\mu\psi
=
\overset{\circ}{D}_{\mu}\psi.
\end{equation*}
When $\tilde{\mathcal L}_{\rm Gr}$ is chosen to be the
Einstein--Hilbert density, this corresponds to the GR geometry, however, 
regularization of the spinor part modifies the matter term to include induced
curvature couplings.

Equation~\eqref{eq:dirac-RC-E} reduces to
\begin{equation}
\left[
i\gamma^\mu\overset{\circ}{D}_{\mu}
-m
-\frac{1}{24}
\mathcal E^{\alpha\beta}\sigma^{\mu\nu}
\overset{\circ}{R}_{\mu\nu\alpha\beta}
\right]\psi
=0.
\label{eq:dirac-LC-E-unreduced}
\end{equation}
For the Levi--Civita curvature tensor,
\begin{equation*}
\overset{\circ}{R}_{[\mu\nu]}=0,
\qquad
\overset{\circ}{P}
:=
\frac12\varepsilon^{\mu\nu\alpha\beta}
\overset{\circ}{R}_{\mu\nu\alpha\beta}
=0,
\end{equation*}
where the second identity follows from the first Bianchi identity.
The curvature--Clifford identity of Appendix~\ref{app:sigmaR}
therefore gives
\begin{equation*}
\sigma^{\alpha\beta}\sigma^{\mu\nu}
\overset{\circ}{R}_{\mu\nu\alpha\beta}
=
2\overset{\circ}{R}.
\end{equation*}
Using \eqref{eq:dual-sigma-star} one likewise obtains
\begin{equation*}
{}^\star\sigma^{\alpha\beta}\sigma^{\mu\nu}
\overset{\circ}{R}_{\mu\nu\alpha\beta}
=
-2i\overset{\circ}{R}\gamma^5.
\end{equation*}
Consequently, \eqref{eq:dirac-LC-E-unreduced} becomes
\begin{equation}
\left[
i\gamma^\mu\overset{\circ}{D}_{\mu}
-m
-\frac{\overset{\circ}{R}}{12}
\left(
\ell-i\ell_5\gamma^5
\right)
\right]\psi
=0.
\label{eq:dirac-LC-E}
\end{equation}

The $\ell$ produces a scalar curvature-dependent mass correction,
whereas $\ell_5$ produces its parity-odd pseudoscalar partner. In
terms of
\begin{equation*}
m_{\rm s}(x)
:=
m+\frac{\ell}{12}\overset{\circ}{R}(x),
\qquad
m_5(x)
:=
\frac{\ell_5}{12}\overset{\circ}{R}(x),
\end{equation*}
the equation takes the form
\begin{equation*}
\left(
i\gamma^\mu\overset{\circ}{D}_{\mu}
-m_{\rm s}
+i m_5\gamma^5
\right)\psi
=0,
\end{equation*}
where $m_{\rm s}$ and $m_5$ are, respectively, scalar and
pseudoscalar spacetime-dependent mass components. In a background
whose curvature varies slowly on the scale of the fermion wavelength,
the corresponding leading-order WKB dispersion relation is
\begin{equation*}
p_\mu p^\mu
=
m_{\rm eff}^2(x),
\qquad
m_{\rm eff}^2(x)
:=
m_{\rm s}^2(x)+m_5^2(x).
\end{equation*}
Spatial variations of $\overset{\circ}{R}$ therefore define a
position-dependent matter-wave mass profile and may produce
refraction-like changes in finite-frequency fermion propagation.

The parity-even coupling
$\overset{\circ}{R}\bar\psi\psi$, which appears in the Dirac equation as
a curvature-dependent mass contribution, was previously derived within
the covariant-canonical framework
\cite{struckmeier2024pauli,vasak2023covariant,vasak23}.
Its cosmological implications were investigated directly in
\cite{benisty2019inflation}, while equivalent nonminimally coupled
fermion and fermion-condensate cosmologies were studied in
\cite{ribas2008cosmological,carloni2014non}.

\section{Characteristics and Causal Bounds}
\label{characteristics}

\subsection{Principal symbol, characteristics, and the Velo--Zwanziger criterion}

We briefly recall the notions needed to formulate the characteristic, or causality,
problem. Let \[Pu=0\] be a square linear system of order $r$ for an
$N$-component field $u$, so that $P$ is an $N\times N$ matrix
of scalar differential operators. In local coordinates
$x=(x^0,\ldots,x^{d-1})$ one may write, in multi-index notation,
\begin{equation}
P
=
\sum_{|\boldsymbol{\alpha}|\le r}
p_{\boldsymbol{\alpha}}(x)\,\partial^{\boldsymbol{\alpha}},
\label{eq:P-multiindex}
\end{equation}
where the coefficients $p_{\boldsymbol{\alpha}}(x)$ are $N\times N$ matrices
acting on the field components, while $\partial^{\boldsymbol{\alpha}}$ acts
componentwise on $u$. Equivalently, $\partial^{\boldsymbol{\alpha}}$ acts as a
scalar differential operator times the identity in field space. Here
\[
\boldsymbol{\alpha}
=(\alpha_0,\ldots,\alpha_{d-1})\in\mathbb{N}_0^{\,d},
\qquad
|\boldsymbol{\alpha}|:=\alpha_0+\cdots+\alpha_{d-1},
\]
and
\begin{equation*}
\partial^{\boldsymbol{\alpha}}
:=\partial_0^{\alpha_0}\cdots\partial_{d-1}^{\alpha_{d-1}} ,
\end{equation*}
where the superscripts on the factors denote the orders of the corresponding
coordinate derivatives.

For a covector \(\xi=\xi_\mu dx^\mu\in T_x^*\mathcal M\), we use the
multi-index power notation
\begin{equation*}
\xi^{\boldsymbol{\alpha}}
:=
(\xi_0)^{\alpha_0}\cdots(\xi_{d-1})^{\alpha_{d-1}} ,
\end{equation*}
where now the superscripts on the factors denote instead powers of the corresponding
covector components.

The principal symbol is read off from the top-order part of \(P\). One
keeps only the terms with \(|\boldsymbol{\alpha}|=r\) and makes the
operator-level substitution
\begin{equation*}
p_{\boldsymbol{\alpha}}(x)\,\partial^{\boldsymbol{\alpha}}
\longmapsto
p_{\boldsymbol{\alpha}}(x)\,\xi^{\boldsymbol{\alpha}} .
\end{equation*}

The principal symbol is then the ordinary \(N\times N\) matrix-valued
homogeneous polynomial
\begin{equation}
\sigma_{P}(x,\xi)
:=
\sum_{|\boldsymbol{\alpha}|=r}
p_{\boldsymbol{\alpha}}(x)\,\xi^{\boldsymbol{\alpha}} .
\label{eq:principal-symbol}
\end{equation}

A nonzero covector $\xi\neq 0$ is said to be characteristic at $x$ if
the symbol matrix $\sigma_{P}(x,\xi)$ is not invertible. Since the
symbol is square, this is equivalent to
\begin{equation}
\det\sigma_{P}(x,\xi)=0.
\label{eq:char-covector}
\end{equation}
Accordingly, a hypersurface $\Sigma:\Phi(x)=0$ is characteristic at
$x\in\Sigma$ if its conormal
\begin{equation*}
n_\mu:=\partial_\mu\Phi
\end{equation*}
satisfies
\begin{equation}
\det\sigma_{P}(x,n)=0;
\label{eq:characteristic-hypersurface}
\end{equation}
otherwise $\Sigma$ is non-characteristic
\cite{hilbert1985methods,taylor1996partial}.

Characteristics matter for two closely related reasons. First, a Cauchy problem
is normally posed on a non-characteristic hypersurface, where the principal part
can be solved locally for the derivatives normal to that hypersurface. On a
characteristic hypersurface this normal evolution map is singular, so standard
Cauchy data cannot in general be prescribed freely, and existence or uniqueness
may fail \cite{vitagliano2014characteristics}. Second, in the wavefront picture,
the boundary of a localized disturbance is a characteristic hypersurface, or
wavefront \cite{hilbert1985methods,vitagliano2014characteristics}.

At a point $x$, the condition $\det\sigma_{P}(x,n)=0$ is homogeneous in the
nonzero conormal $n_\mu$, hence it defines a cone in $T_x^*\mathcal M$, the
characteristic cone \cite{hilbert1985methods,taylor1996partial}. In Lorentzian
problems this cone plays the role of an effective causal structure for
high-frequency front propagation: for the wave operator the characteristics are
null hypersurfaces, \textit{i.e.} the light cone \cite{hilbert1985methods}.

For a first-order system ($r=1$), Eq.~\eqref{eq:P-multiindex} contains
one-derivative terms with $|\boldsymbol{\alpha}|=1$ and a zeroth-order term
with $\boldsymbol{\alpha}=(0,\ldots,0)$. Since
$\partial^{(0,\ldots,0)}$ is the identity operator, the latter is simply a
matrix multiplication term, which we collect into $B(x)$. The one-derivative terms
determine the principal symbol. For example, in four dimensions,
\[
p_{(1,0,0,0)}(x)\partial_0
\longmapsto
p_{(1,0,0,0)}(x)\xi_0 ,
\]
with analogous terms for the other coordinate directions. Collecting the
one-derivative coefficient matrices into $A^\mu(x)$, the operator has the form
\[
P=A^\mu(x)\partial_\mu+B(x),
\]
and hence
\begin{equation}
\sigma_{P}(x,\xi)=A^\mu(x)\xi_\mu .
\label{eq:first-order-symbol}
\end{equation}

For the standard Dirac operator in particular, \eqref{eq:first-order-symbol}
gives $\sigma_D(x,n)=i\gamma^\mu n_\mu$, and therefore
\begin{equation}
\det(i\gamma^\mu n_\mu)=(n^2)^2,
\qquad
n^2:=g^{\mu\nu}n_\mu n_\nu ,
\label{eq:standard-dirac-characteristics}
\end{equation}
so Dirac wavefronts are also null \cite{vitagliano2014characteristics}.

To connect the conormal $n_\mu$ to a front speed, consider locally
$\Sigma:\Phi(t,\vec x)=0$ with
\[
n_\mu=\partial_\mu\Phi=(n_0,\vec n)
\]
in an inertial chart of signature $(+,-,-,-)$. If the front moves in the
direction of its spatial normal, then along $\Sigma$
\begin{equation*}
0=\frac{d\Phi}{dt}
=n_0+\vec n\cdot\frac{d\vec x}{dt}.
\end{equation*}
Up to the irrelevant orientation of the normal, this gives
\begin{equation*}
v_{\rm front}=\frac{|n_0|}{|\vec n|},
\qquad
n^2=n_0^2-|\vec n|^2
=|\vec n|^2\left(v_{\rm front}^2-1\right).
\end{equation*}
Hence $n^2>0$, namely a timelike conormal, implies
$v_{\rm front}>1$, \textit{i.e.} a superluminal wavefront, while causal fronts
relative to the background metric require $n^2\le 0$.

Operationally, for first-order systems one can extract the characteristic
equation via the Hadamard discontinuity, or Velo--Zwanziger shock, method
\cite{velo1969propagation,hilbert1985methods,fabbri2018non}. One assumes that
$u$ is continuous across $\Sigma$, but its first derivatives have a jump of the
form
\begin{equation*}
[\partial_\mu u]_\Sigma
:=
(\partial_\mu u)_+\big|_\Sigma
-
(\partial_\mu u)_-\big|_\Sigma
=n_\mu a,
\end{equation*}
where $a$ is a nonzero field-space discontinuity amplitude. Keeping only the
top-derivative terms in the field equation then yields
\begin{equation*}
\sigma_{P}(x,n)a=0,
\end{equation*}
so $a$ lies in the kernel of the symbol matrix. A nontrivial wavefront can
therefore exist only if
\begin{equation}
\det\sigma_{P}(x,n)=0.
\label{eq:hadamard-characteristic-determinant}
\end{equation}

In the Velo--Zwanziger reading one fixes a real spatial covector $\vec n$ and
views the characteristic equation as an equation for $n_0$: complex roots signal
a failure of the necessary real-root condition for hyperbolicity, while real
roots with $n^2>0$ describe acausal, superluminal characteristic fronts
\cite{fabbri2018non}. The criterion is used here as an exclusion test: complex
roots or timelike characteristic covectors establish a pathology, whereas their
absence does not by itself prove strong hyperbolicity or complete well-posedness.

The preceding discussion also fixes which induced terms are relevant for the
characteristic analysis. The principal symbol is sensitive only to the
Clifford-valued coefficient of the highest derivative. Connections appearing
inside \(D_\mu\) (\textit{e.g.} electromagnetic, weak, color) implement
local transformations of the corresponding internal or Lorentz-frame
components, after the highest derivative is isolated, their connection
coefficients multiply \(\psi\) algebraically. The same is true of masses, field
strengths produced by commutators, and curvature corrections. Such terms may
affect finite-wavelength propagation, spectra, or effective masses, but they do
not change the wavefront cone. A change of characteristics requires a change in
the matrix multiplying the conormal in the principal symbol. In the present
metric-affine construction this change is carried only by the torsion and
nonmetricity pieces of \(\Gamma^\mu_{\rm MA}\).

\subsection{Metric-affine characteristic determinant}
\label{subsec:metric-affine-characteristic-determinant}

We now apply the preceding characteristic criterion to the full
metric-affine Dirac equation
\begin{equation}
\left(
i\Gamma^\beta_{\rm MA}D_\beta-M_{MA}
\right)\psi=0 ,
\label{eq:GenDirac-characteristics}
\end{equation}
where \(M_{MA}\) denotes the lower-order matrix in
Eq.~\eqref{eq:M-metricaffine}, and
\begin{equation}
\Gamma^\beta_{\rm MA}
=
\gamma^\beta
-
\frac{1}{3}
\left(
\frac{1}{e}\tilde{\mathcal E}^{\alpha\beta}{}_{;\alpha}
-
2\mathcal E^{\xi\beta}S^\alpha{}_{\xi\alpha}
+
\mathcal E^{\alpha\xi}S^\beta{}_{\alpha\xi}
\right).
\label{eq:Gamma-characteristics}
\end{equation}
Since
\begin{equation*}
D_\beta=\partial_\beta+\Omega_\beta ,
\end{equation*}
while \(\Omega_\beta\) and \(M_{MA}\) contain no derivatives of \(\psi\), only
\(\Gamma^\beta_{\rm MA}\partial_\beta\psi\) contributes to the principal
symbol. Thus
\begin{equation}
\sigma(x,\xi)=i\Gamma^\beta_{\rm MA}(x)\xi_\beta .
\label{eq:principal-symbol-characteristics}
\end{equation}
The overall factor of \(i\) does not affect the zero set of the determinant,
so for a conormal \(n_\beta\) the characteristic equation is
\begin{equation}
\det\!\left(\mathbb M\right)=0, \quad \mathbb M:=\Gamma^\beta_{\rm MA}n_\beta .
\label{eq:characteristic-equation-general}
\end{equation}
The characteristic problem is local, so all background quantities are understood
to be frozen at the point under consideration.

Define
\begin{equation}
\mathfrak X^\beta
:=
\frac{1}{e}\tilde{\mathcal E}^{\alpha\beta}{}_{;\alpha}
-
2\mathcal E^{\xi\beta}S^\alpha{}_{\xi\alpha}
+
\mathcal E^{\alpha\xi}S^\beta{}_{\alpha\xi},
\label{eq:Xfrak-def}
\end{equation}
so that
\begin{equation*}
\Gamma^\beta_{\rm MA}
=
\gamma^\beta-\frac13\mathfrak X^\beta .
\end{equation*}
Therefore the characteristic matrix can be written as
\begin{equation}
\mathbb M
=
n_\beta\gamma^\beta
-
\frac13 n_\beta\mathfrak X^\beta .
\label{eq:M-def}
\end{equation}

The Clifford structure of \(\mathbb M\) is simpler than the general Dirac
basis decomposition 
\begin{equation} 
\mathbb M = \mathcal S + i\mathcal P\gamma^5 + \mathcal V_\mu\gamma^\mu + \mathcal A_\mu\gamma^\mu\gamma^5 + \frac12\mathcal T_{\mu\nu}\sigma^{\mu\nu} 
\label{eq:M-dirac-basis} 
\end{equation}
would suggest. From \eqref{eq:dual-sigma-star} and
\eqref{eq:E-metricaffine}, the bivector \(\mathcal E^{\mu\nu}\) lies in the
tensor sector of the Clifford algebra. The same is true of
\(\mathfrak X^\beta\): the covariant derivative in
\(\tilde{\mathcal E}^{\alpha\beta}{}_{;\alpha}\) and the torsion contractions
in \eqref{eq:Xfrak-def} do not generate scalar, pseudoscalar, vector, or
axial-vector Clifford components. Thus we may write
\begin{equation}
\mathfrak X^\beta
=
\frac12 X^\beta{}_{\mu\nu}\sigma^{\mu\nu},
\qquad
X^\beta{}_{\mu\nu}=-X^\beta{}_{\nu\mu}.
\label{eq:X-tensor-expansion}
\end{equation}
Substitution into \eqref{eq:M-def} gives
\begin{equation*}
\mathbb M
=
n_\mu\gamma^\mu
-
\frac16 n_\beta X^\beta{}_{\mu\nu}\sigma^{\mu\nu}.
\end{equation*}
Equivalently,
\begin{equation}
\mathbb M
=
n_\mu\gamma^\mu
+
\frac12\mathcal T_{\mu\nu}\sigma^{\mu\nu},
\label{eq:M-vector-tensor-form}
\end{equation}
where
\begin{equation}
\mathcal T_{\mu\nu}
=
-\frac13 n_\beta X^\beta{}_{\mu\nu},
\qquad
\mathcal T_{\mu\nu}=-\mathcal T_{\nu\mu}.
\label{eq:T-X}
\end{equation}

We now use the standard determinant formula for a Dirac matrix of the
vector--tensor form \eqref{eq:M-vector-tensor-form}. Define
\begin{equation}
({}^\star\mathcal T)_{\mu\nu}
:=
\frac12\varepsilon_{\mu\nu}{}^{\rho\sigma}
\mathcal T_{\rho\sigma},
\label{eq:T-dual}
\end{equation}
and
\begin{equation}
\mathcal T_{\pm\mu\nu}
:=
\frac12
\left(
\mathcal T_{\mu\nu}
\pm i({}^\star\mathcal T)_{\mu\nu}
\right),
\qquad
({}^\star\mathcal T_\pm)_{\mu\nu}
=
\mp i\mathcal T_{\pm\mu\nu}.
\label{eq:T-pm}
\end{equation}
In a chiral representation, with
\begin{equation*}
\sigma^\mu=(\mathbf 1,\sigma^j),
\qquad
\bar\sigma^\mu=(\mathbf 1,-\sigma^j),
\end{equation*}
\eqref{eq:M-vector-tensor-form} admits the $2\times2$ block decomposition
\begin{equation}
\mathbb M
=
\begin{pmatrix}
\frac{i}{2}\mathcal T_{-\mu\nu}\sigma^\mu\bar\sigma^\nu
&
n_\mu\sigma^\mu
\\[4pt]
n_\mu\bar\sigma^\mu
&
\frac{i}{2}\mathcal T_{+\mu\nu}\bar\sigma^\mu\sigma^\nu
\end{pmatrix}.
\label{eq:M-chiral-block}
\end{equation}
The determinant of a $2\times2$ block matrix with $2\times2$ blocks can be written as
\begin{align*}
\det\!\begin{pmatrix}A&B\\ C&D\end{pmatrix}
=\det(AD)+\det(BC)-\mathrm{tr}\!\big(\bar B\,A\,\bar C\,D\big),
\\
\bar B:=\mathrm{adj}(B),\quad \bar C:=\mathrm{adj}(C),
\end{align*}
where $\mathrm{adj}(\cdot)$ denotes the matrix adjoint. This yields the closed form expression~\cite{kostelecky2013fermions,red2008parametrization}:
\begin{equation}
\det\mathbb M
=
(n^2)^2
+
(\mathcal T_-^2)(\mathcal T_+^2)
-
8n_\mu\mathcal T_-^{\mu\nu}
\mathcal T_{+\nu\rho}n^\rho .
\label{eq:det-vector-tensor}
\end{equation}
Thus the general metric-affine characteristic equation is
\begin{equation}
(n^2)^2
+
(\mathcal T_-^2)(\mathcal T_+^2)
-
8n_\mu\mathcal T_-^{\mu\nu}
\mathcal T_{+\nu\rho}n^\rho
=
0.
\label{eq:characteristic-final}
\end{equation}

It remains to simplify the dependence on the two regularization parameters.
By Legendre regularity,
\begin{equation*}
\ell^2+\ell_5^2\neq0 .
\end{equation*}
We may therefore define
\begin{equation}
\lambda:=\sqrt{\ell^2+\ell_5^2},
\qquad
\cos\theta:=\frac{\ell}{\lambda},
\qquad
\sin\theta:=\frac{\ell_5}{\lambda}.
\label{eq:lambda-theta}
\end{equation}
Then
\begin{equation}
\ell-i\ell_5\gamma^5
=
\lambda e^{-i\theta\gamma^5}.
\label{eq:ell-chiral-phase}
\end{equation}
Since
\begin{equation}
{}^\star\sigma^{\mu\nu}
=
-i\sigma^{\mu\nu}\gamma^5,
\qquad
[\sigma^{\mu\nu},\gamma^5]=0,
\label{eq:sigma-gamma5-identities}
\end{equation}
the regularizing bivector may be written as
\begin{equation}
\mathcal E^{\mu\nu}
=
\lambda\sigma^{\mu\nu}e^{-i\theta\gamma^5}.
\label{eq:E-chiral-phase}
\end{equation}
Now we define
\begin{equation}
R_\theta
:=
\exp\!\left(-\frac{i}{2}\theta\gamma^5\right),
\label{eq:Rtheta}
\end{equation}
then, since \(\sigma^{\mu\nu}\) commutes with \(\gamma^5\),
\begin{equation}
R_\theta\sigma^{\mu\nu}R_\theta
=
\sigma^{\mu\nu}R_\theta^2
=
\sigma^{\mu\nu}e^{-i\theta\gamma^5}.
\label{eq:Rtheta-sigma}
\end{equation}
On the other hand, since \(\gamma^\mu\) anticommutes with \(\gamma^5\),
\begin{equation*}
R_\theta\gamma^\mu
=
\gamma^\mu R_\theta^{-1},
\end{equation*}
and therefore
\begin{equation}
R_\theta\gamma^\mu R_\theta
=
\gamma^\mu .
\label{eq:Rtheta-gamma}
\end{equation}
Equations \eqref{eq:E-chiral-phase}--\eqref{eq:Rtheta-gamma} give
\begin{equation}
\mathcal E^{\mu\nu}
=
\lambda R_\theta\sigma^{\mu\nu}R_\theta ,
\label{eq:E-chiral-rotation}
\end{equation}
while the vector Dirac matrix \(\gamma^\mu\) is unchanged by the same
double-sided operation.

The same relation also holds after the metric-affine differentiation
appearing in \(\mathfrak X^\beta\). Indeed, \(R_\theta\) is built only from
\(\gamma^5\), and is thus covariantly constant. Since \(R_\theta\) carries no spacetime
indices, we have
\begin{equation}
\left(
R_\theta A^{\mu\nu}R_\theta
\right)_{;\alpha}
=
R_\theta A^{\mu\nu}{}_{;\alpha}R_\theta
\label{eq:covder-Rtheta}
\end{equation}
for any Clifford-valued bivector \(A^{\mu\nu}\). Therefore
\begin{equation}
\frac1e
\tilde{\mathcal E}^{\alpha\beta}{}_{;\alpha}(\ell,\ell_5)
=
R_\theta
\left[
\frac1e
\tilde{\mathcal E}^{\alpha\beta}{}_{;\alpha}(\lambda,0)
\right]
R_\theta .
\label{eq:E-density-chiral}
\end{equation}
The two torsion terms in \(\mathfrak X^\beta\) are algebraic in
\(\mathcal E^{\mu\nu}\), so they transform in the same way. Hence
\begin{equation}
\mathfrak X^\beta(\ell,\ell_5)
=
R_\theta
\mathfrak X^\beta(\lambda,0)
R_\theta .
\label{eq:X-chiral}
\end{equation}
Together with \eqref{eq:Rtheta-gamma}, this yields
\begin{equation}
\Gamma^\beta_{\rm MA}(\ell,\ell_5)
=
R_\theta
\Gamma^\beta_{\rm MA}(\lambda,0)
R_\theta ,
\label{eq:Gamma-chiral}
\end{equation}
and, after contraction with \(n_\beta\),
\begin{equation}
\mathbb M(\ell,\ell_5)
=
R_\theta
\mathbb M(\lambda,0)
R_\theta .
\label{eq:M-chiral}
\end{equation}

The determinant is unchanged, because
\begin{equation}
\det R_\theta
=
\exp\!\left[
-\frac{i}{2}\theta\,\operatorname{tr}(\gamma^5)
\right]
=
1.
\label{eq:det-Rtheta}
\end{equation}
Therefore
\begin{equation}
\det\mathbb M(\ell,\ell_5)
=
\det\mathbb M(\lambda,0)
.
\label{eq:det-lambda-reduction}
\end{equation}
Thus the two-parameter characteristic equation can be obtained by working in
the parity-even representative
\begin{equation}
\mathcal E^{\mu\nu}
=
\lambda\sigma^{\mu\nu}
\label{eq:E-lambda-sector-convention}
\end{equation}
and then interpreting \(\lambda\) as the invariant length
\(\sqrt{\ell^2+\ell_5^2}\) of the original pair \((\ell,\ell_5)\). In
particular, any characteristic bound derived in the \(\ell_5=0\) sector is
promoted to the full two-parameter theory by the replacement
\begin{equation}
\ell^2
\longmapsto
\lambda^2
=
\ell^2+\ell_5^2 .
\label{eq:ell-to-lambda-characteristics}
\end{equation}
The distinction between \(\ell\) and \(\ell_5\) remains present in lower-order
terms and in parity-odd dipole couplings, but it does not affect the
principal characteristic determinant. The following subsections use this construction for convenience.

\subsection{Riemann--Cartan sectors}
\label{subsec:RC-torsion-VZ}

We first specialize to the metric-compatible torsionful limit,
\begin{equation*}
Q_{\lambda\mu\nu}=0,
\qquad
S^\lambda{}_{\mu\nu}\neq0 .
\end{equation*}
Using the convention \eqref{eq:E-lambda-sector-convention}, the principal
coefficient becomes
\begin{equation}
\Gamma^\beta_{\rm RC}
=
\gamma^\beta
+
\frac{\lambda}{3}
\left(
2\sigma^{\xi\beta}S^\alpha{}_{\xi\alpha}
-
\sigma^{\xi\alpha}S^\beta{}_{\xi\alpha}
\right).
\label{eq:Gamma-RC-lambda}
\end{equation}
The characteristic equation is
\begin{equation}
\det\!\left(\Gamma^\mu_{\rm RC}n_\mu\right)=0 .
\label{eq:char-RC-det}
\end{equation}

We decompose torsion into its Lorentz-irreducible pieces,
\begin{equation*}
t_\mu:=S^\alpha{}_{\mu\alpha},
\qquad
s^\mu:=-\frac16\varepsilon^{\mu\nu\rho\sigma}S_{\nu\rho\sigma},
\end{equation*}
and
\begin{equation*}
q^\alpha{}_{\mu\alpha}=0,
\qquad
\varepsilon^{\mu\nu\rho\sigma}q_{\nu\rho\sigma}=0 .
\end{equation*}
Then, the decomposition is
\begin{equation}
S^\beta{}_{\xi\alpha}
=
\frac23\delta^\beta{}_{[\alpha}t_{\xi]}
+
\varepsilon^\beta{}_{\xi\alpha\lambda}s^\lambda
+
q^\beta{}_{\xi\alpha}.
\label{eq:torsion-irrep-RC}
\end{equation}
Substitution into \eqref{eq:Gamma-RC-lambda} gives
\begin{equation}
\Gamma^\beta_{\rm RC}
=
\gamma^\beta
+
\frac{4\lambda}{9}\sigma^{\xi\beta}t_\xi
-
\frac{\lambda}{3}\sigma^{\xi\alpha}
\varepsilon^\beta{}_{\xi\alpha\lambda}s^\lambda
-
\frac{\lambda}{3}\sigma^{\xi\alpha}q^\beta{}_{\xi\alpha}.
\label{eq:Gamma-RC-irrep}
\end{equation}
Thus
\begin{align}
\Gamma^\mu_{\rm RC}n_\mu
&=
n_\mu\gamma^\mu
+
\frac{4\lambda}{9}\sigma^{\nu\mu}t_\nu n_\mu
-
\frac{\lambda}{3}\sigma^{\nu\rho}
\varepsilon^\mu{}_{\nu\rho\sigma}s^\sigma n_\mu
\nonumber\\
&\quad
-
\frac{\lambda}{3}\sigma^{\nu\rho}q^\mu{}_{\nu\rho}n_\mu .
\label{eq:Gamma-n-RC-irrep}
\end{align}
Comparing with \eqref{eq:M-vector-tensor-form}, one obtains
\begin{equation}
\mathcal T_{\mu\nu}^{\rm RC}
=
\frac{8\lambda}{9}t_{[\mu}n_{\nu]}
-
\frac{2\lambda}{3}
\varepsilon^\rho{}_{\mu\nu\alpha}n_\rho s^\alpha
-
\frac{2\lambda}{3}n_\rho q^\rho{}_{\mu\nu}.
\label{eq:T-RC-irrep}
\end{equation}
Together with \eqref{eq:det-vector-tensor}, this gives the general
Riemann--Cartan characteristic determinant in irreducible torsion variables.

\subsubsection{Pure axial torsion}
\label{subsubsec:RC-pure-axial-VZ}

For purely axial torsion,
\begin{equation*}
t_\mu=0,
\qquad
q^\lambda{}_{\mu\nu}=0,
\end{equation*}
the tensor coefficient reduces to
\begin{equation*}
\mathcal T_{\mu\nu}^{(s)}
=
-\frac{2\lambda}{3}
\varepsilon^\rho{}_{\mu\nu\alpha}n_\rho s^\alpha .
\end{equation*}
The determinant becomes
\begin{equation}
\det\!\left(\Gamma^\mu_{(s)}n_\mu\right)
=
\frac{1}{81}
\left[
9n^2
+
4\lambda^2
\left(
n^2s^2-(n\cdot s)^2
\right)
\right]^2 .
\label{eq:det-RC-axial}
\end{equation}
Thus the characteristic equation is
\begin{equation}
9n^2
+
4\lambda^2
\left(
n^2s^2-(n\cdot s)^2
\right)
=
0,
\label{eq:char-RC-axial}
\end{equation}
or equivalently
\begin{equation}
\left(9+4\lambda^2s^2\right)n^2
=
4\lambda^2(n\cdot s)^2 .
\label{eq:char-RC-axial-rearranged}
\end{equation}

\begin{itemize}
\item \textit{Timelike axial torsion.}
If \(s^\mu\) is timelike, choose its rest frame,
\begin{equation*}
s^\mu=(s_0,0,0,0).
\end{equation*}
Then \eqref{eq:char-RC-axial-rearranged} gives
\begin{equation*}
n_0^2
=
\left(
1+\frac{4\lambda^2s_0^2}{9}
\right)|\vec n|^2 .
\end{equation*}
Hence the roots are real, but
\begin{equation*}
\frac{n_0^2}{|\vec n|^2}>1
\end{equation*}
for \(s_0\neq0\). Timelike axial torsion is therefore acausal.

\item \textit{Lightlike axial torsion.}
If \(s^\mu\) is lightlike, choose
\begin{equation*}
s^\mu=S(1,1,0,0).
\end{equation*}
For the spatial direction parallel to the spatial part of \(s^\mu\), take
\begin{equation*}
n_\mu=(n_0,n_1,0,0).
\end{equation*}
Then \eqref{eq:char-RC-axial} becomes
\begin{equation*}
(n_0^2-n_1^2)
=
\alpha_s(n_0+n_1)^2,
\quad
\alpha_s:=\frac{4\lambda^2S^2}{9}.
\end{equation*}
One root is null, \(n_0=-n_1\). The other root is
\begin{equation*}
\frac{n_0}{n_1}
=
\frac{1+\alpha_s}{1-\alpha_s},
\end{equation*}
whenever \(\alpha_s\neq1\), and has absolute value larger than one. Thus
the lightlike axial sector admits timelike characteristic covectors. At the
exceptional value \(\alpha_s=1\), transverse spatial directions give a
degenerate or non-hyperbolic characteristic polynomial. Hence the lightlike
axial sector is not admissible.

\item \textit{Spacelike axial torsion.}
If \(s^\mu\) is spacelike, choose
\begin{equation*}
s^\mu=(0,s_1,0,0).
\end{equation*}
Equation \eqref{eq:char-RC-axial-rearranged} gives
\begin{equation*}
n_0^2
=
n_2^2+n_3^2
+
\frac{9}{9-4\lambda^2s_1^2}\,n_1^2 .
\end{equation*}
For
\begin{equation*}
9-4\lambda^2s_1^2>0,
\qquad
\text{i.e.}
\qquad
|s_1|<\frac{3}{2\lambda},
\end{equation*}
the roots are real for all spatial directions, but propagation parallel to
\(\vec s\) gives
\begin{equation*}
\frac{n_0^2}{|\vec n|^2}
=
\frac{9}{9-4\lambda^2s_1^2}>1
\end{equation*}
for \(s_1\neq0\). Thus this regime is hyperbolic but acausal. For
\[
9-4\lambda^2s_1^2<0,
\]
the same parallel spatial direction gives imaginary \(n_0\), so the
Velo--Zwanziger real-root test fails. The threshold
\(9-4\lambda^2s_1^2=0\) is a degenerate limiting case and is not an
admissible hyperbolic regime.
\end{itemize}

Therefore every nonzero purely axial torsion background is excluded. This result also applies to the already mentioned quadratic
nonminimal fermion coupling in the (axial) Poincaré-gauge construction of \cite{obukhov2014spin}.

\subsubsection{Pure trace torsion}
\label{subsubsec:RC-pure-trace-VZ}

For purely trace torsion,
\begin{equation*}
s^\mu=0,
\qquad
q^\lambda{}_{\mu\nu}=0,
\end{equation*}
one has
\begin{equation*}
\mathcal T_{\mu\nu}^{(\mathrm{tr})}
=
\frac{8\lambda}{9}t_{[\mu}n_{\nu]} .
\end{equation*}
The determinant reduces to
\begin{equation}
\det\!\left(\Gamma^\mu_{(\mathrm{tr})}n_\mu\right)
=
\left[
n^2
+
\frac{16\lambda^2}{81}
\left(
n^2t^2-(n\cdot t)^2
\right)
\right]^2 .
\label{eq:det-RC-trace}
\end{equation}
Thus the characteristic equation is
\begin{equation}
\left(
1+\frac{16\lambda^2}{81}t^2
\right)n^2
=
\frac{16\lambda^2}{81}(n\cdot t)^2 .
\label{eq:char-RC-trace}
\end{equation}

\begin{itemize}
\item \textit{Timelike trace torsion.}
If \(t^\mu\) is timelike, choose
\begin{equation*}
t^\mu=(t_0,0,0,0).
\end{equation*}
Then
\begin{equation*}
n_0^2
=
\left(
1+\frac{16\lambda^2t_0^2}{81}
\right)|\vec n|^2 ,
\end{equation*}
so the characteristic roots are real but superluminal for \(t_0\neq0\).

\item \textit{Lightlike trace torsion.}
If \(t^\mu\) is lightlike, choose
\begin{equation*}
t^\mu=T(1,1,0,0).
\end{equation*}
For \(n_\mu=(n_0,n_1,0,0)\), Eq.~\eqref{eq:char-RC-trace} gives
\begin{equation*}
(n_0^2-n_1^2)
=
\alpha_t(n_0+n_1)^2,
\qquad
\alpha_t:=\frac{16\lambda^2T^2}{81}.
\end{equation*}
Apart from the null root \(n_0=-n_1\), the second root is
\begin{equation*}
\frac{n_0}{n_1}
=
\frac{1+\alpha_t}{1-\alpha_t},
\end{equation*}
for \(\alpha_t\neq1\), and is timelike. At the exceptional value
\(\alpha_t=1\), the characteristic polynomial degenerates or fails to give
real roots in transverse directions. Thus the lightlike trace sector is not
admissible.

\item \textit{Spacelike trace torsion.}
If \(t^\mu\) is spacelike, choose
\begin{equation*}
t^\mu=(0,t_1,0,0).
\end{equation*}
Equation \eqref{eq:char-RC-trace} gives
\begin{equation*}
n_0^2
=
n_2^2+n_3^2
+
\frac{1}{1-\frac{16\lambda^2t_1^2}{81}}\,n_1^2 .
\end{equation*}
For
\begin{equation*}
1-\frac{16\lambda^2t_1^2}{81}>0,
\qquad
\text{i.e.}
\qquad
|t_1|<\frac{9}{4\lambda},
\end{equation*}
the roots are real, but propagation parallel to \(\vec t\) gives
\begin{equation*}
\frac{n_0^2}{|\vec n|^2}
=
\frac{1}{1-\frac{16\lambda^2t_1^2}{81}}>1
\end{equation*}
for \(t_1\neq0\). For
\[
1-\frac{16\lambda^2t_1^2}{81}<0,
\]
the parallel spatial direction gives imaginary roots. The equality case is
again degenerate.
\end{itemize}

Therefore every nonzero purely trace torsion background is excluded.

\subsubsection{Pure tensor torsion}
\label{subsubsec:RC-pure-tensor-VZ}

For purely tensorial torsion,
\begin{equation*}
t_\mu=0,
\qquad
s^\mu=0,
\end{equation*}
the tensor coefficient is
\begin{equation*}
\mathcal T_{\mu\nu}^{(q)}
=
-\frac{2\lambda}{3}n_\rho q^\rho{}_{\mu\nu}.
\end{equation*}
It is useful to define
\begin{equation}
F_{\mu\nu}:=n^\alpha q_{\alpha\mu\nu},
\qquad
w_\mu:=n^\alpha n^\beta q_{\alpha\beta\mu}.
\label{eq:F-w-q-def}
\end{equation}
Then
\begin{equation*}
\mathcal T_{\mu\nu}^{(q)}
=
-\frac{2\lambda}{3}F_{\mu\nu}.
\end{equation*}
With
\begin{equation*}
F^2:=F_{\mu\nu}F^{\mu\nu},
\quad
F\cdot{}^\star F
:=
F_{\mu\nu}({}^\star F)^{\mu\nu},
\quad
w^2:=w_\mu w^\mu,
\end{equation*}
the characteristic determinant becomes
\begin{equation}
0
=
\left(
n^2-\frac{2\lambda^2}{9}F^2
\right)^2
+
\frac{4\lambda^4}{81}
\left(F\cdot{}^\star F\right)^2
+
\frac{16\lambda^2}{9}w^2 .
\label{eq:det-RC-tensor}
\end{equation}
Unlike the axial and trace sectors, this equation does not reduce to a
single-vector deformed cone. The quantities \(F_{\mu\nu}\) and \(w_\mu\)
depend on the characteristic covector \(n_\mu\) itself, so a universal
classification requires additional algebraic assumptions on
\(q^\lambda{}_{\mu\nu}\). No general causal classification is attempted here.

\subsubsection{Mixed axial and trace torsion}
\label{subsubsec:RC-mixed-axial-trace-VZ}

We now consider the mixed axial--trace sector,
\begin{equation*}
q^\lambda{}_{\mu\nu}=0,
\qquad
s^\mu\neq0,
\qquad
t^\mu\neq0 .
\end{equation*}
The determinant may be written directly in terms of the invariants
\begin{align}
X_s
&:=
n^2s^2-(n\cdot s)^2,
&
X_t
&:=
n^2t^2-(n\cdot t)^2,
\nonumber\\
X_{st}
&:=
n^2(s\cdot t)-(n\cdot s)(n\cdot t).
\label{eq:Xs-Xt-Xst}
\end{align}
Using \eqref{eq:T-RC-irrep}, one obtains
\begin{align}
\det\!\left(\Gamma^\mu_{(st)}n_\mu\right)
&=
\left[
n^2
+
\frac{4\lambda^2}{9}X_s
+
\frac{16\lambda^2}{81}X_t
\right]^2
\nonumber\\ &\quad+
\frac{256\lambda^4}{729}
\left(
X_{st}^2-X_sX_t
\right).
\label{eq:det-RC-mixed-st}
\end{align}

For the causal analysis it is cleaner to absorb the two numerical
coefficients into rescaled background vectors,
\begin{equation}
a^\mu:=\frac{2\lambda}{3}s^\mu,
\qquad
b^\mu:=\frac{4\lambda}{9}t^\mu .
\label{eq:a-b-scaled}
\end{equation}
Define
\begin{align}
&Y_a
:=
n^2a^2-(n\cdot a)^2,
\quad
Y_b
:=
n^2b^2-(n\cdot b)^2,
\nonumber\\
&Y_{ab}
:=
n^2(a\cdot b)-(n\cdot a)(n\cdot b).
\label{eq:Yab-def}
\end{align}
Then \eqref{eq:det-RC-mixed-st} becomes
\begin{equation}
\det\!\left(\Gamma^\mu_{(st)}n_\mu\right)
=
\left[
n^2+Y_a+Y_b
\right]^2
+
4\left(
Y_{ab}^2-Y_aY_b
\right).
\label{eq:det-RC-mixed-ab}
\end{equation}

If \(a^\mu\) and \(b^\mu\) are linearly dependent, the determinant reduces to
a one-vector equation of the same type as the pure axial or pure trace
sector. The previous one-vector analysis then excludes every nonzero
linearly dependent mixed background.

It remains to consider the case in which \(a^\mu\) and \(b^\mu\) span a
two-dimensional plane
\begin{equation*}
\Pi:=\operatorname{span}\{a,b\}.
\end{equation*}
The determinant \eqref{eq:det-RC-mixed-ab} depends on the pair
\((a,b)\) only through the trace and determinant of the Gram matrix induced
on \(\Pi\). Therefore an ordinary \(O(2)\) rotation of the pair
\((a,b)\) leaves the determinant unchanged. For a nondegenerate plane, we
may choose an orthogonal pair adapted to the Lorentz signature of \(\Pi\).

\begin{itemize}
\item \textit{Timelike plane.}
If \(\Pi\) is timelike, choose a local frame such that
\begin{equation*}
a^\mu=(A,0,0,0),
\qquad
b^\mu=(0,B,0,0),
\qquad
A>0.
\end{equation*}
For the spatial direction \(n_\mu=(n_0,n_1,0,0)\), the determinant becomes
\begin{equation}
\det\!\left(\Gamma^\mu_{(st)}n_\mu\right)
=
\left[
(1+A^2)n_1^2+(B^2-1)n_0^2
\right]^2 .
\label{eq:mixed-timelike-plane-det}
\end{equation}
If \(B^2<1\), the characteristic root is
\begin{equation*}
n_0^2
=
\frac{1+A^2}{1-B^2}\,n_1^2>n_1^2,
\end{equation*}
so timelike characteristic covectors occur. If \(B^2>1\), the same equation
gives \(n_0^2<0\), so the real-root test fails. If \(B^2=1\), choose
instead \(n_\mu=(n_0,0,n_2,0)\). Then one obtains the root
\begin{equation*}
n_0^2
=
\left(1+\frac{A^2}{4}\right)n_2^2>n_2^2 .
\end{equation*}
Thus every timelike mixed plane is inadmissible.

\item \textit{Spacelike plane.}
If \(\Pi\) is spacelike, choose an orthogonal frame and relabel the two
basis vectors so that
\begin{equation*}
a^\mu=(0,A,0,0),
\qquad
b^\mu=(0,0,B,0),
\qquad
B\ge A>0.
\end{equation*}
For \(n_\mu=(n_0,n_1,0,0)\), the determinant factorizes as
\begin{align}
\det\!\left(\Gamma^\mu_{(st)}n_\mu\right)
&=
\left[
(B-1)^2n_1^2+
\left(A^2-(B-1)^2\right)n_0^2
\right]
\nonumber\\
&\quad\times
\left[
(B+1)^2n_1^2+
\left(A^2-(B+1)^2\right)n_0^2
\right].
\label{eq:mixed-spacelike-plane-det}
\end{align}
Since \(B+1>A\), the second factor gives
\begin{equation*}
n_0^2
=
\frac{(B+1)^2}{(B+1)^2-A^2}\,n_1^2>n_1^2 .
\end{equation*}
Therefore every spacelike mixed plane admits timelike characteristic
covectors and is inadmissible.

\item \textit{Lightlike plane.}
Finally suppose that \(\Pi\) is degenerate but two-dimensional. Choose
\begin{equation*}
a^\mu=A(1,1,0,0),
\quad
b^\mu=(0,0,B,0),
\quad
A>0,
\quad
B>0.
\end{equation*}
For \(n_\mu=(n_0,n_1,0,0)\), the determinant factorizes as
\begin{align}
\det\!\left(\Gamma^\mu_{(st)}n_\mu\right)
&=
(n_0+n_1)^2
\nonumber\\
&\quad\times
\left[
A^2(n_0+n_1)+(B-1)^2(n_1-n_0)
\right]
\nonumber\\
&\quad\times
\left[
A^2(n_0+n_1)+(B+1)^2(n_1-n_0)
\right].
\label{eq:mixed-null-plane-det}
\end{align}
Apart from the null root \(n_0=-n_1\), the two remaining factors give roots
of the form
\begin{equation*}
\frac{n_0}{n_1}
=
-\frac{A^2+(B\mp1)^2}{A^2-(B\mp1)^2},
\end{equation*}
whenever the corresponding denominator is nonzero. Such roots have absolute
value larger than one and are therefore timelike. If the denominator in one
factor vanishes, the other factor gives the timelike root, except at the
isolated case \(B=1\), \(A=2\). In that case, taking instead
\(n_\mu=(n_0,0,n_2,0)\) gives a factor proportional to
\begin{equation*}
n_2^2+24n_0^2,
\end{equation*}
so the characteristic root is imaginary. Hence the degenerate null plane is
also inadmissible.
\end{itemize}

Combining the linearly dependent, timelike-plane, spacelike-plane, and
null-plane cases, we conclude that the mixed axial--trace sector is excluded
for every nonzero pair \((s^\mu,t^\mu)\). Therefore, within the
Riemann--Cartan axial--trace subsector, the only physically admissible
background is
\begin{equation*}
s^\mu=0,
\qquad
t^\mu=0 .
\end{equation*}

\subsection{Torsionless nonmetricity sectors}
\label{subsec:N-nonmetricity-VZ}

We next specialize to the torsionless nonmetric limit,
\begin{equation*}
S^\lambda{}_{\mu\nu}=0,
\qquad
Q_{\lambda\mu\nu}\neq0 .
\end{equation*}
Using \eqref{eq:E-lambda-sector-convention}, the principal coefficient
\eqref{eq:Gamma-N-E} becomes
\begin{equation}
\Gamma^\beta_{\rm N}
=
\gamma^\beta
-
\frac{\lambda}{6}
\left[
(\widehat Q_\lambda-Q_\lambda)\sigma^{\lambda\beta}
+
Q_\alpha{}^\beta{}_\lambda\sigma^{\alpha\lambda}
\right],
\label{eq:Gamma-N-lambda}
\end{equation}
where the two traces \(Q_\lambda\) and \(\widehat Q_\lambda\) are those
defined in Eq.~\eqref{eq:Q-traces-N}. The characteristic equation is
\begin{equation}
\det\!\left(\Gamma^\mu_{\rm N}n_\mu\right)=0 .
\label{eq:char-N-det}
\end{equation}
Comparing with the sector form \eqref{eq:M-vector-tensor-form}, one obtains
\begin{equation}
\mathcal T_{\mu\nu}^{\rm N}
=
-\frac{\lambda}{3}
(\widehat Q-Q)_{[\mu}n_{\nu]}
-
\frac{\lambda}{3}
n_\rho Q_{[\mu}{}^\rho{}_{\nu]} .
\label{eq:T-N-general}
\end{equation}

For later use, we recall the standard trace decomposition of nonmetricity in
four dimensions,
\begin{equation}
Q_{\lambda\mu\nu}
=
Q^{(W)}_{\lambda\mu\nu}
+
Q^{(t)}_{\lambda\mu\nu}
+
N_{\lambda\mu\nu},
\label{eq:N-trace-decomposition}
\end{equation}
where
\begin{equation}
Q^{(W)}_{\lambda\mu\nu}
=
\frac14 Q_\lambda g_{\mu\nu},
\label{eq:N-Weyl-piece}
\end{equation}
and
\begin{equation}
Q^{(t)}_{\lambda\mu\nu}
=
\frac49 g_{\lambda(\mu}\Lambda_{\nu)}
-
\frac19\Lambda_\lambda g_{\mu\nu},
\quad
\Lambda_\lambda
:=
\widehat Q_\lambda-\frac14Q_\lambda .
\label{eq:N-second-trace-piece}
\end{equation}
The tensor \(N_{\lambda\mu\nu}\) is tracefree,
\begin{equation}
N_{\lambda\mu}{}^\mu=0,
\qquad
N^\mu{}_{\mu\lambda}=0.
\label{eq:N-tracefree-conditions}
\end{equation}

\subsubsection{Pure Weyl nonmetricity}
\label{subsubsec:N-pure-Weyl-VZ}

For pure Weyl nonmetricity,
\begin{equation*}
Q_{\lambda\mu\nu}
=
Q^{(W)}_{\lambda\mu\nu}
=
\frac14 Q_\lambda g_{\mu\nu},
\quad
Q^{(t)}_{\lambda\mu\nu}=0,
\quad
N_{\lambda\mu\nu}=0.
\end{equation*}
Then
\begin{equation*}
\widehat Q_\lambda=\frac14Q_\lambda,
\qquad
Q_\alpha{}^\beta{}_\lambda
=
\frac14Q_\alpha\delta^\beta{}_\lambda .
\end{equation*}
Hence
\begin{align*}
(\widehat Q_\lambda-Q_\lambda)\sigma^{\lambda\beta}
+
Q_\alpha{}^\beta{}_\lambda\sigma^{\alpha\lambda}
&=
-\frac34Q_\lambda\sigma^{\lambda\beta}
+
\frac14Q_\alpha\delta^\beta{}_\lambda\sigma^{\alpha\lambda}
\nonumber\\
&=
-\frac12Q_\lambda\sigma^{\lambda\beta}.
\end{align*}
Therefore
\begin{equation}
\mathcal T_{\mu\nu}^{(W)}
=
\frac{\lambda}{6}Q_{[\mu}n_{\nu]}
=
\frac{\lambda}{12}
\left(
Q_\mu n_\nu-Q_\nu n_\mu
\right).
\label{eq:T-N-Weyl}
\end{equation}

Since \(\mathcal T_{\mu\nu}^{(W)}\) is a simple bivector, the determinant
reduces to
\begin{equation}
\det\!\left(\Gamma^\mu_{(W)}n_\mu\right)
=
\left[
n^2
+
\frac{\lambda^2}{144}
\left(
Q^2n^2-(Q\cdot n)^2
\right)
\right]^2 .
\label{eq:det-N-Weyl}
\end{equation}
Thus the characteristic equation in the pure Weyl sector is
\begin{equation}
\left(
1+\frac{\lambda^2}{144}Q^2
\right)n^2
=
\frac{\lambda^2}{144}(Q\cdot n)^2 .
\label{eq:char-N-Weyl}
\end{equation}

\begin{itemize}
\item \textit{Timelike Weyl vector.}
If \(Q^\mu\) is timelike, choose its rest frame,
\begin{equation*}
Q^\mu=(Q_0,0,0,0).
\end{equation*}
Then \eqref{eq:char-N-Weyl} gives
\begin{equation*}
n_0^2
=
\left(
1+\frac{\lambda^2Q_0^2}{144}
\right)|\vec n|^2 .
\end{equation*}
The roots are real, but timelike characteristic covectors occur for
\(Q_0\neq0\). Thus timelike Weyl nonmetricity is acausal by the
Velo--Zwanziger criterion.

\item \textit{Lightlike Weyl vector.}
If \(Q^\mu\) is lightlike, choose
\begin{equation*}
Q^\mu=Q_L(1,1,0,0).
\end{equation*}
For \(n_\mu=(n_0,n_1,0,0)\), Eq.~\eqref{eq:char-N-Weyl} becomes
\begin{equation*}
n_0^2-n_1^2
=
\alpha_Q(n_0+n_1)^2,
\qquad
\alpha_Q:=\frac{\lambda^2Q_L^2}{144}.
\end{equation*}
Apart from the null root \(n_0=-n_1\), the second root is
\begin{equation*}
\frac{n_0}{n_1}
=
\frac{1+\alpha_Q}{1-\alpha_Q},
\end{equation*}
for \(\alpha_Q\neq1\), and has absolute value larger than one. At the
exceptional value \(\alpha_Q=1\), transverse spatial directions fail the
real-root test. Hence lightlike Weyl nonmetricity is not admissible.

\item \textit{Spacelike Weyl vector.}
If \(Q^\mu\) is spacelike, choose
\begin{equation*}
Q^\mu=(0,Q_1,0,0).
\end{equation*}
Then \eqref{eq:char-N-Weyl} gives
\begin{equation*}
n_0^2
=
n_2^2+n_3^2
+
\frac{144}{144-\lambda^2Q_1^2}\,n_1^2 .
\end{equation*}
For
\begin{equation*}
144-\lambda^2Q_1^2>0,
\qquad
\text{i.e.}
\qquad
|Q_1|<\frac{12}{\lambda},
\end{equation*}
the roots are real for all spatial directions, but propagation parallel to
\(\vec Q\) gives timelike characteristic covectors whenever \(Q_1\neq0\).
For
\[
144-\lambda^2Q_1^2<0,
\]
the same parallel spatial direction gives imaginary \(n_0\). The equality
case is degenerate and is not an admissible hyperbolic regime.
\end{itemize}

Therefore every nonzero pure Weyl nonmetricity background is excluded by the
Velo--Zwanziger front-causality criterion.

\subsubsection{Pure second-trace nonmetricity}
\label{subsubsec:N-pure-second-trace-VZ}

For pure second-trace nonmetricity, we set
\begin{align*}
&Q_{\lambda\mu\nu}
=
Q^{(t)}_{\lambda\mu\nu}
=
\frac49 g_{\lambda(\mu}\Lambda_{\nu)}
-
\frac19\Lambda_\lambda g_{\mu\nu},
\\\
&Q^{(W)}_{\lambda\mu\nu}=0,
\qquad
N_{\lambda\mu\nu}=0.
\end{align*}
Then
\begin{equation*}
Q_\lambda=0,
\qquad
\widehat Q_\lambda=\Lambda_\lambda,
\end{equation*}
and
\begin{equation*}
Q_\alpha{}^\beta{}_\lambda
=
\frac29
\left(
\delta^\beta{}_\alpha\Lambda_\lambda
+
g_{\alpha\lambda}\Lambda^\beta
\right)
-
\frac19\Lambda_\alpha\delta^\beta{}_\lambda .
\end{equation*}
Thus
\begin{align*}
(\widehat Q_\lambda-Q_\lambda)\sigma^{\lambda\beta}
+
Q_\alpha{}^\beta{}_\lambda\sigma^{\alpha\lambda}
&=
\Lambda_\lambda\sigma^{\lambda\beta}
+
\frac29\delta^\beta{}_\alpha\Lambda_\lambda\sigma^{\alpha\lambda}
\nonumber\\&\quad+
\frac29g_{\alpha\lambda}\Lambda^\beta\sigma^{\alpha\lambda}
-\frac19\Lambda_\alpha\delta^\beta{}_\lambda\sigma^{\alpha\lambda}
\nonumber\\
&=
\Lambda_\lambda\sigma^{\lambda\beta}
-
\frac29\Lambda_\lambda\sigma^{\lambda\beta}
-
\frac19\Lambda_\lambda\sigma^{\lambda\beta}
\nonumber\\
&=
\frac23\Lambda_\lambda\sigma^{\lambda\beta}.
\end{align*}
Therefore
\begin{equation}
\mathcal T_{\mu\nu}^{(t)}
=
-\frac{2\lambda}{9}\Lambda_{[\mu}n_{\nu]}
=
-\frac{\lambda}{9}
\left(
\Lambda_\mu n_\nu-\Lambda_\nu n_\mu
\right).
\label{eq:T-N-second-trace}
\end{equation}

Again the tensor coefficient is a simple bivector, so the determinant becomes
\begin{equation}
\det\!\left(\Gamma^\mu_{(t)}n_\mu\right)
=
\left[
n^2
+
\frac{\lambda^2}{81}
\left(
\Lambda^2n^2-(\Lambda\cdot n)^2
\right)
\right]^2 .
\label{eq:det-N-second-trace}
\end{equation}
Hence the characteristic equation in the pure second-trace sector is
\begin{equation}
\left(
1+\frac{\lambda^2}{81}\Lambda^2
\right)n^2
=
\frac{\lambda^2}{81}(\Lambda\cdot n)^2 .
\label{eq:char-N-second-trace}
\end{equation}

The causal analysis is identical to the pure Weyl case, with the replacement
\begin{equation*}
Q_\mu\longmapsto\Lambda_\mu,
\qquad
\frac{\lambda^2}{144}\longmapsto\frac{\lambda^2}{81}.
\end{equation*}
Explicitly, for a timelike second-trace vector
\(\Lambda^\mu=(\Lambda_0,0,0,0)\), one obtains
\begin{equation*}
n_0^2
=
\left(
1+\frac{\lambda^2\Lambda_0^2}{81}
\right)|\vec n|^2 ,
\end{equation*}
so timelike characteristic covectors occur whenever \(\Lambda_0\neq0\). For
a lightlike vector \(\Lambda^\mu=\Lambda_L(1,1,0,0)\), the non-null root is
\begin{equation*}
\frac{n_0}{n_1}
=
\frac{1+\alpha_\Lambda}{1-\alpha_\Lambda},
\qquad
\alpha_\Lambda:=\frac{\lambda^2\Lambda_L^2}{81},
\end{equation*}
and is timelike for \(\alpha_\Lambda\neq1\), while the exceptional value
fails the real-root test in transverse directions. Finally, for a spacelike
vector \(\Lambda^\mu=(0,\Lambda_1,0,0)\), one finds
\begin{equation*}
n_0^2
=
n_2^2+n_3^2
+
\frac{81}{81-\lambda^2\Lambda_1^2}\,n_1^2 .
\end{equation*}
Thus the regime
\begin{equation*}
|\Lambda_1|<\frac{9}{\lambda}
\end{equation*}
has real roots but timelike characteristic covectors for propagation parallel
to \(\vec\Lambda\), whereas
\[
|\Lambda_1|>\frac{9}{\lambda}
\]
violates the real-root criterion in the same direction. The equality case is
degenerate.

Therefore every nonzero pure second-trace nonmetricity background is excluded
by the Velo--Zwanziger front-causality criterion.

\subsubsection{Combined trace-vector sector}
\label{subsubsec:N-combined-trace-VZ}

Although the pure Weyl and pure second-trace sectors have different numerical
coefficients when written in terms of \(Q_\mu\) and \(\Lambda_\mu\), their
combination is not a genuinely two-vector characteristic problem. For the
trace-vector sector
\begin{equation*}
Q_{\lambda\mu\nu}
=
Q^{(W)}_{\lambda\mu\nu}
+
Q^{(t)}_{\lambda\mu\nu},
\qquad
N_{\lambda\mu\nu}=0,
\end{equation*}
define the effective trace vector
\begin{equation}
H_\mu
:=
\widehat Q_\mu-Q_\mu
=
\Lambda_\mu-\frac34Q_\mu .
\label{eq:H-effective-nonmetric-trace}
\end{equation}
Combining the Weyl and second-trace contractions gives
\begin{equation*}
(\widehat Q_\lambda-Q_\lambda)\sigma^{\lambda\beta}
+
Q_\alpha{}^\beta{}_\lambda\sigma^{\alpha\lambda}
=
\frac23H_\lambda\sigma^{\lambda\beta}.
\end{equation*}
Therefore
\begin{equation}
\mathcal T_{\mu\nu}^{\rm tr}
=
-\frac{2\lambda}{9}H_{[\mu}n_{\nu]} .
\label{eq:T-N-trace-vector}
\end{equation}
The determinant is
\begin{equation}
\det\!\left(\Gamma^\mu_{\rm tr}n_\mu\right)
=
\left[
n^2
+
\frac{\lambda^2}{81}
\left(
H^2n^2-(H\cdot n)^2
\right)
\right]^2 .
\label{eq:det-N-trace-vector}
\end{equation}
Thus
\begin{equation*}
H_\mu\neq0
\end{equation*}
is excluded by the same one-vector Velo--Zwanziger argument as above, while
\begin{equation}
H_\mu=0,
\qquad
\text{i.e.}
\qquad
\widehat Q_\mu=Q_\mu,
\label{eq:H-zero-condition}
\end{equation}
leaves the principal characteristic cone undeformed:
\begin{equation*}
\det\!\left(\Gamma^\mu_{\rm tr}n_\mu\right)
=
(n^2)^2 .
\end{equation*}

\subsubsection{Pure tracefree nonmetricity}
\label{subsubsec:N-tracefree-VZ}

We finally consider a pure tracefree nonmetricity background,
\begin{equation*}
Q_{\lambda\mu\nu}=N_{\lambda\mu\nu},
\qquad
Q_\lambda=0,
\qquad
\widehat Q_\lambda=0.
\end{equation*}
Then we have
\begin{equation}
\mathcal T_{\mu\nu}^{(N)}
=
-\frac{\lambda}{3}
n_\rho N_{[\mu}{}^\rho{}_{\nu]} .
\label{eq:T-N-tracefree}
\end{equation}

It is useful to define
\begin{equation}
F_{\mu\nu}
:=
n_\rho N_{[\mu}{}^\rho{}_{\nu]},
\qquad
w_\mu
:=
n^\rho F_{\rho\mu}.
\label{eq:F-w-N-def}
\end{equation}
Then
\begin{equation*}
\mathcal T_{\mu\nu}^{(N)}
=
-\frac{\lambda}{3}F_{\mu\nu}.
\end{equation*}
Substitution into the general vector--tensor determinant gives the diagnostic
formula
\begin{equation}
\det\!\left(\Gamma^\mu_{(N)}n_\mu\right)
=
\left(
n^2-\frac{\lambda^2}{18}F^2
\right)^2
+
\frac{\lambda^4}{324}
\left(
F\cdot{}^\star F
\right)^2
+
\frac{4\lambda^2}{9}w^2 ,
\label{eq:det-N-tracefree}
\end{equation}
where
\begin{equation*}
F^2:=F_{\mu\nu}F^{\mu\nu},
\quad
F\cdot{}^\star F:=F_{\mu\nu}({}^\star F)^{\mu\nu},
\quad
w^2:=w_\mu w^\mu .
\end{equation*}
The characteristic equation is obtained by setting
\eqref{eq:det-N-tracefree} to zero. Yet, unlike the Weyl and second-trace sectors, this expression does not reduce to
a single-vector deformed cone. The quantities \(F_{\mu\nu}\) and \(w_\mu\)
depend on the characteristic covector \(n_\mu\) itself, and \(w^2\) is a
Lorentzian square. Therefore no universal causal classification can be attempted without further algebraic
assumptions on \(N_{\lambda\mu\nu}\).

There is, however, one simple subcase. If the pure tracefree nonmetricity is
totally symmetric,
\begin{equation*}
N_{\lambda\mu\nu}=N_{(\lambda\mu\nu)},
\qquad
N_{\lambda\mu}{}^\mu=0,
\end{equation*}
then
\begin{equation*}
N_{[\mu}{}^\rho{}_{\nu]}=0,
\qquad
F_{\mu\nu}=0,
\end{equation*}
and hence
\begin{equation}
\Gamma^\beta_{(N)}=\gamma^\beta,
\qquad
\det\!\left(\Gamma^\mu_{(N)}n_\mu\right)
=
(n^2)^2 .
\end{equation}
Thus the totally symmetric tracefree subcase is invisible to the principal
symbol and is not excluded by the Velo--Zwanziger front-causality criterion
at this level.

The main results of the Velo--Zwanziger analyses of this section are summarized in Table \eqref{tab:VZ-summary}.

\begin{table*}[t]
\caption{
Summary of the Velo--Zwanziger test for the analyzed torsion and nonmetricity sectors.
``Fails'' means that every nonzero background in the sector either admits timelike
characteristic covectors or fails the real-root test. ``Passes'' means only that the
sector is not immediately excluded by the Velo--Zwanziger criterion.
}
\label{tab:VZ-summary}
\centering
\renewcommand{\arraystretch}{1.35}
\begin{ruledtabular}
\begin{tabular}{lll}
Name & Object & V--Z Test \\
\hline
Pure axial torsion
&
\(s^\mu\)
&
Fails for \(s^\mu\neq0\).
\\
Pure trace torsion
&
\(t^\mu\)
&
Fails for \(t^\mu\neq0\).
\\
Mixed axial--trace torsion
&
\((s^\mu,t^\mu)\)
&
Fails for \((s^\mu,t^\mu)\neq(0,0)\).
\\
Pure tensor torsion
&
\(q^\lambda{}_{\mu\nu}\)
&
Inconclusive in general.
\\
Pure Weyl nonmetricity
&
\(Q_\mu\)
&
Fails for \(Q_\mu\neq0\).
\\
Pure second-trace nonmetricity
&
\(\Lambda_\mu\)
&
Fails for \(\Lambda_\mu\neq0\).
\\
Combined trace-vector nonmetricity
&
\(H_\mu=\widehat Q_\mu-Q_\mu\)
&
Fails for \(H_\mu\neq0\), passes if \(H_\mu=0\).
\\
Pure tracefree nonmetricity
&
\(N_{\lambda\mu\nu}\)
&
Inconclusive in general, passes if \(N_{\lambda\mu\nu}=N_{(\lambda\mu\nu)}\).
\end{tabular}
\end{ruledtabular}
\end{table*}

\section{Discussion and Conclusions}
\label{sec:discussion-conclusions}

The central question addressed here was not whether a known quadratic surface term can render the Dirac Lagrangian nondegenerate, but which such terms are allowed once regularity and symmetry are imposed simultaneously. Starting from the singular covariant Legendre map of the free Dirac theory, we derived the Hessian-relevant null deformations that leave the free flat-space Dirac equation unchanged while making the derivative--momentum map locally invertible:
\[
\Delta\mathcal L=\frac{i}{3}
(\partial_\mu\bar\psi)
\left(
\ell\sigma^{\mu\nu}
+
\ell_5{}^\star\sigma^{\mu\nu}
\right)
(\partial_\nu\psi),
\qquad
\ell^2+\ell_5^2\neq0 .
\]
Thus the Gasiorowicz term proportional to \(\ell\) is recovered as the
parity-even member \(\ell_5=0\) of a two-parameter proper-Lorentz-covariant
null class. If parity or CP invariance is imposed, this class reduces to the
representative used in \cite{struckmeier2024pauli}.

After minimal \(U(1)\) gauging, the free null term is no longer a divergence,
since the gauge-covariant derivatives do not commute. The parity-even
parameter \(\ell\) produces the usual Pauli-type magnetic-dipole operator,
while the parity-odd parameter \(\ell_5\) produces its electric-dipole-type
dual. The additional Maxwell source is a dipole current of superpotential
form: for localized fields it carries no net \(U(1)\) charge, but modifies
the local polarization--magnetization structure of the source. In this sense,
Legendre regularity gives rise to precisely the Pauli-type dipole couplings.

These dipole descendants also make the regularization lengths
experimentally constrainable. For each fermion species \(f\), magnetic
and electric dipole measurements bound the same pair
\((\ell_f,\ell_{5f})\) that defines the selected free representative. For
the electron, the present magnetic-moment uncertainty gives an
experimental sensitivity
\[
        |\ell_e|\lesssim 3.8\times10^{-10}\,{\rm GeV}^{-1},
\]
while the electron EDM bound gives
\[
        |\ell_{5e}|\lesssim 6.2\times10^{-16}\,{\rm GeV}^{-1}.
\]
These estimates show that the parity-odd branch is already constrained
very strongly by EDM searches, while the parity-even branch is probed by
precision magnetic-moment measurements.

The value of the coupling parameters for different flavors of particles may
additionally be constrained by the symmetries of electroweak theory. In the
unbroken electroweak theory, the left-handed charged leptons and their
corresponding left-handed neutrinos form $SU(2)_L$ doublets,
$L_i=(\nu_{iL},\ell_{iL})^T$, with $i=e,\mu,\tau$. In order for this symmetry
to remain intact, the two components of a given doublet may not be assigned
independent regularization parameters. Thus $\nu_{eL}$ and $e_L$,
$\nu_{\mu L}$ and $\mu_L$, and $\nu_{\tau L}$ and $\tau_L$ must each share the
same parameter pair within the corresponding left-handed doublet. The
right-handed charged leptons are $SU(2)_L$ singlets and are therefore not tied
to these doublet parameters by $SU(2)_L$ alone. If right-handed neutrinos are
included, they are likewise electroweak singlets and may carry independent
parameters \cite{denk2019weak}.

The same mechanism persists after spacetime gauging. In the Lorentz-spinor
metric-affine prescription used here, the regularized representative induces
spin-curvature, torsion, and nonmetricity contributions to the Dirac
operator. In the Riemann--Cartan limit and for \(\ell_5=0\), the purely
axial torsion sector reproduces, at the level of operator structure, the
Pauli-type spin-torsion and spin-curvature terms postulated by Obukhov
\textit{et al.} in a nonminimal coupling \cite{obukhov2014spin}, up to a choice of coefficients.

The characteristic analysis gives a more restrictive conclusion than \cite{struckmeier2024pauli}. The determinant depends on the pair
\((\ell,\ell_5)\) only through
\[
  \lambda=\sqrt{\ell^2+\ell_5^2},
\]
because the principal symbol is related to its \(\ell_5=0\) term
by an exact double-sided chiral transformation with unit determinant. Thus
\(\ell_5\) remains visible in lower-order parity-odd terms, but not in the
front cone. The Velo--Zwanziger results are summarized in
Table~\ref{tab:VZ-summary}. In particular, the pure axial torsion, pure trace
torsion, and mixed axial--trace torsion sectors all fail for every nonzero
background. Hence the apparent causal windows initially found in \cite{struckmeier2024pauli} for the axial and vector torsion cases are
eliminated by this analysis. Within the
Riemann--Cartan axial--trace subsector of the induced operator, the
Velo--Zwanziger admissible point is \(s^\mu=t^\mu=0\) -- this would be particularly troubling for theories with propagating torsion such as the one discussed in \cite{denk2023torsion}. The tensor irreducible
torsion component is not excluded in general, as its determinant depends on
covector-dependent contractions and requires further algebraic assumptions.

The torsionless nonmetricity sectors show an analogous pattern. Pure Weyl
and pure second-trace nonmetricity each reduce to a one-vector deformation of
the light cone and fail for every nonzero vector. In the combined trace
sector, only the effective vector $H_\mu=\widehat Q_\mu-Q_\mu$
enters the principal symbol. \(H_\mu\neq0\) fails, while \(H_\mu=0\) leaves
the principal cone undeformed. The tracefree nonmetricity sector, like tensor
torsion, is not universally classified. A totally symmetric tracefree
nonmetricity tensor is invisible to the principal symbol and is therefore
not excluded at this level.

We recall, however, that in the Levi--Civita limit all induced modifications of
the principal part vanish, while the Legendre map remains regular for
$\ell^2+\ell_5^2\neq0$. Only lower-order scalar and pseudoscalar
curvature-dependent mass terms survive, and the characteristic cone
remains the standard metric null cone. The regularized theory is
therefore not excluded by the Velo--Zwanziger criterion in this sector.

Our torsion results may also be compared with the Velo--Zwanziger analysis of
\cite{fabbri2019restrictions}. There, derivative couplings of spinors to
the axial and trace torsion vectors are introduced phenomenologically and are
then restricted by the Velo--Zwanziger criterion. The tensor irreducible part
of torsion is not included in the subsequent spinor analysis, since it is
argued there not to define a consistent propagating field. In the present
work, however, this component is not excluded a priori, though one cannot generally classify it.

The full determinant in \cite{fabbri2019restrictions} is approximated by
weak-background and hierarchy assumptions. After these reductions, the extra
scalar and pseudoscalar derivative couplings appear only through effective
combinations of coupling constants. The resulting reduced cone is therefore
structurally the same as the weak-background limit of our pure axial and pure
trace determinants. Our analysis differs in that the Legendre-induced
principal symbol is treated exactly in the pure axial, pure trace, and mixed
axial--trace sectors. It also applies the same Velo--Zwanziger diagnostic to
selected torsionless nonmetricity sectors, which are not treated in \cite{fabbri2019restrictions}.

The static form of the resulting Dirac equation is also naturally comparable
with the fermion sector of the Standard-Model Extension. The SME is an
effective field theory in which fixed background coefficients control
Lorentz- and CPT-violating operators. Its physical interpretation is
therefore different from ours. Nevertheless, once the metric-affine
background is frozen locally, the Dirac equation derived here has the same
kind of Clifford-valued operator structure: scalar, pseudoscalar, vector,
axial-vector, and tensor terms may appear in the lower-order operator, while
the derivative part is modified by tensor-valued coefficients built from
torsion or nonmetricity. In this sense the formalism resembles a restricted,
geometrically induced SME-type fermion sector, with the crucial difference
that the coefficients are not arbitrary EFT parameters but are fixed by
\(\lambda\) and by the local non-Riemannian geometry.

This comparison is consistent with earlier SME analyses of torsion. In
particular, constant torsion backgrounds have been interpreted as effective
sources of local Lorentz violation and matched to SME fermion coefficients
\cite{shapiro2002physical,kostelecky2008constraints}. The present construction
differs in origin: the torsion and nonmetricity couplings are not postulated as or restricted
to independent Lorentz-violating background operators, they arise from
spacetime gauging of the Legendre-regular Dirac Lagrangian.

The construction also has a precise interpretation in the language of
Lepagean regularization. In \cite{krupkova2001legendre}, singular
first-order Lagrangians are regularized by replacing the
Poincar\'e--Cartan form \(\theta_L\) with a Lepagean equivalent
\[
  \rho=\theta_L+\nu ,
\]
where \(\nu\) is an at least \(2\)-contact correction. In the
Hamilton \(p_2\)-regularizations relevant here, the chosen correction is
represented as the \(2\)-contact part of an associated \(n\)-form
\(\eta\). This changes the Hamiltonian side of the theory, but not the
Euler--Lagrange equations. The horizontal part \(h(\eta)\) associated
with this construction is called the satellite, and the modified
Lagrangian
\[
  \bar{L}=L-h(\eta)
\]
is the corresponding so-called dedonderization of \(L\). When the
associated form \(\eta\) is closed, \(\bar L\) is variationally
equivalent to \(L\) and regular in the ordinary Hamilton--De Donder
sense. The generalized Hamilton equations based on \(\rho\) then
coincide with the usual Hamilton--De Donder equations of \(\bar L\).

The free Dirac term constructed here fits this pattern. Up to the sign
convention in the definition of \(h(\eta)\), the derivative-quadratic null
density
\[
  \Delta\mathcal L
  =
  \frac{i}{3}
  (\partial_\mu\bar\psi)
  \mathcal E^{\mu\nu}
  (\partial_\nu\psi)
\]
is a satellite term, and \(L_D+\Delta L\) is a dedonderized Dirac
Lagrangian. The general Lepagean construction allows a large regularizing
freedom, and the Dirac symmetries used here reduce that freedom to the
spinorial bivector
\[
  \mathcal E^{\mu\nu}
  =
  \ell\sigma^{\mu\nu}
  +
  \ell_5{}^\star\sigma^{\mu\nu}.
\]
Thus the Lepagean viewpoint explains the Hamiltonian regularity mechanism,
while the present analysis identifies the physically admissible Dirac
satellite and studies what it produces after gauging. Broad Lepagean
constructions are developed for example in \cite{krupkova2009lepage}.

This distinction also clarifies the non-Riemannian principal-symbol
obstructions found above. There are two possible uses of the same algebraic bivector
\(\mathcal E^{\mu\nu}\). The route followed in the main text, and in
\cite{struckmeier2024pauli}, is to take the dedonderized horizontal
Lagrangian \(L_D+\Delta L\) as the matter Lagrangian and then apply minimal
coupling. Since covariant derivatives do not commute, the satellite is no
longer variationally trivial after gauging. It produces the Pauli-type
electromagnetic, spin-torsion, spin-curvature, and nonmetricity terms
analyzed here. The Velo--Zwanziger restrictions derived above apply to this
physical-gauging interpretation.

A genuinely Lepagean alternative would keep the regularizing correction as
part of the Lepage form, rather than first converting it into a horizontal
satellite Lagrangian and then gauging that representative. In such a route,
one would gauge the underlying Dirac matter at the level of the Lepage equivalent, and the
regularizing term should be retained as a contact correction to the Hamilton form.
Because the correction remains invariant, it does
not acquire the curvature-commutator terms that arise from gauging the
horizontal satellite, thus not generating new physical coupling. This achieves a Legendre-regular gauged theory without altering the standard gauged field equations.  

This suggests a possible route for covariant Hamiltonian formulations
of fermions in Riemann--Cartan or metric-affine gravity. If the goal is to obtain a regular Hamilton--De Donder
fermion theory, a Lepagean regularization may provide a path to do so that avoids the
propagation pathologies found in the non-Riemannian sectors -- at the cost of losing the induced dipole moments. 

\section*{Acknowledgments}

M.N.M. gratefully acknowledges support from the Margarete und Herbert Puschmann-Stiftung. D.V. and V.D. gratefully acknowledge support from the Fueck-Stiftung.

\appendix
\section{Geometric intuition for fiberwise regularity}
\label{app:visualizing-regularity}

The condition in Eq.~\eqref{eq:regularity-condition} can be viewed geometrically as
a curvature condition on the graph of the Lagrangian over the derivative
fiber. At fixed spacetime point and fixed field value, collect the derivative
variables as
\begin{equation}
  v^I:=v^A{}_\mu,
  \qquad
  I=(A,\mu).
\end{equation}
The relevant graph is
\begin{equation}
  z=\mathcal L(v^1,\ldots,v^N),
\end{equation}
where only the derivative variables are varied. This is not the graph of
\(\mathcal L\) over all variables, but only over the derivative fiber.

For this \(N\)-dimensional graph, the determinant-type curvature is the
Gauss--Kronecker curvature,
\begin{equation}
  K_{\mathrm{GK}}
  =
  \frac{\det W}
       {\left(1+|\nabla_v\mathcal L|^2\right)^{(N+2)/2}},
  \label{eq:gauss-kronecker-curvature}
\end{equation}
where \(W\) is the fiber Hessian of Eq.~\eqref{eq:fiber-hessian}
\cite{Chen2012}. Hence
\begin{equation}
  K_{\mathrm{GK}}\neq0
  \qquad\Longleftrightarrow\qquad
  \det W\neq0 .
\end{equation}
Thus, in this graph picture, Legendre regularity means nondegenerate
curvature in the derivative directions. The slopes of the graph are the
polymomenta, and the curvature measures whether these slopes change
nondegenerately when the derivatives are varied.

To make this visible, consider the simple two-dimensional case. Denote two
derivative coordinates by
\begin{equation}
  v_1:=v^{A_1}{}_{\mu_1},
  \qquad
  v_2:=v^{A_2}{}_{\mu_2}.
\end{equation}
The fiber Hessian is
\begin{equation}
  W_{(1,2)}
  =
  \begin{pmatrix}
    \displaystyle
    \frac{\partial^2\mathcal L}{\partial v_1^2}
    &
    \displaystyle
    \frac{\partial^2\mathcal L}{\partial v_1\partial v_2}
    \\[1em]
    \displaystyle
    \frac{\partial^2\mathcal L}{\partial v_2\partial v_1}
    &
    \displaystyle
    \frac{\partial^2\mathcal L}{\partial v_2^2}
  \end{pmatrix},
\end{equation}
where partial differentiation already implies that all other variables are
kept fixed.

For the graph \(z=\mathcal L(v_1,v_2)\), the Gaussian curvature is
\begin{equation}
  K_{(1,2)}
  =
  \frac{\det W_{(1,2)}}
       {\left(1+(\mathcal L_{,1})^{\,2}
                +(\mathcal L_{,2})^{\,2}\right)^2},
  \label{eq:gaussian-curvature-section}
\end{equation}
where \(\mathcal L_{,i}:=\partial\mathcal L/\partial v_i\)
\cite{doCarmo1976,Spivak1979}. Hence, in a two-dimensional derivative
section,
\begin{equation}
  K_{(1,2)}\neq0
  \qquad\Longleftrightarrow\qquad
  \det W_{(1,2)}\neq0 .
\end{equation}

\paragraph*{Dirac-like affine section.}
The closest simple model of the standard Dirac derivative dependence is
\begin{equation}
  \mathcal L_{\mathrm{aff}}
  =
  a_1v_1+a_2v_2+c .
\end{equation}
Then
\begin{equation}
  W_{(1,2)}=0.
\end{equation}
The graph is a plane, and the momenta are
\begin{equation}
  p_1=a_1,
  \qquad
  p_2=a_2.
\end{equation}
Thus neither derivative can be recovered from the momenta. This is the visual
analogue of a first-order Lagrangian that is linear in the derivative
variables.

\paragraph*{Curved but singular section.}
A section may be curved and still singular. For example,
\begin{equation}
  \mathcal L_{\mathrm{cyl}}
  =
  a_1v_1+a_2v_2+c+\frac{\alpha}{2}v_1^2,
  \qquad
  \alpha\neq0 ,
\end{equation}
has
\begin{equation}
  W_{(1,2)}
  =
  \begin{pmatrix}
    \alpha & 0\\
    0 & 0
  \end{pmatrix},
  \qquad
  \det W_{(1,2)}=0 .
\end{equation}
The graph is a parabolic cylinder. It bends in the \(v_1\)-direction but is
flat in the \(v_2\)-direction. Thus visible curvature in one direction is not
enough: regularity requires no null direction in the fiber Hessian.

\paragraph*{Regular saddle section.}
A regular section need not be convex. Consider
\begin{equation}
  \mathcal L_{\mathrm{sad}}
  =
  a_1v_1+a_2v_2+c+\lambda v_1v_2,
  \qquad
  \lambda\neq0 .
\end{equation}
Then
\begin{equation}
  W_{(1,2)}
  =
  \begin{pmatrix}
    0 & \lambda\\
    \lambda & 0
  \end{pmatrix},
  \qquad
  \det W_{(1,2)}=-\lambda^2\neq0 .
\end{equation}
The graph is a saddle. The Hessian has eigenvalues of opposite sign, so it is
indefinite but nondegenerate.

\paragraph*{Higher-order section.}
Higher-order derivative terms may affect regularity away from a point, but
they cannot repair a degenerate Hessian at that point. For instance,
\begin{equation}
  \mathcal L_{\mathrm{ho}}
  =
  a_1v_1+a_2v_2+c
  +
  \frac{\kappa}{4}\left(v_1^4+v_2^4\right),
  \qquad
  \kappa\neq0 ,
\end{equation}
gives
\begin{equation}
  W_{(1,2)}
  =
  \begin{pmatrix}
    3\kappa v_1^2 & 0\\
    0 & 3\kappa v_2^2
  \end{pmatrix},
  \qquad
  \det W_{(1,2)}
  =
  9\kappa^2v_1^2v_2^2 .
\end{equation}
This determinant can be nonzero away from the coordinate axes, but at the
origin
\begin{equation}
  W_{(1,2)}(0,0)=0.
\end{equation}
Such a term therefore does not regularize the Hessian pointwise.

\section{Obukhov \textit{et al.} (2014): Nonminimal Fermionic Couplings in Poincar\'e Gauge Gravity}
\label{app-obkhv}

Obukhov \textit{et al.}~\cite{obukhov2014spin} work in the
Riemann--Cartan spacetime of Poincar\'e gauge gravity, whose covariant
gravitational field strengths are the torsion \(S_{ij}{}^\alpha\) and the
curvature \(R_{ij}{}^{\alpha\beta}\). Their nonminimal Pauli ansatz uses
the translational and Lorentz gravitational moment two-forms obtained
from the Gordon decomposition of the Dirac energy--momentum and spin
currents \cite{Hehl:1997currents}. Besides the convective parts, those
currents contain polarizational total-divergence contributions, which
define gravitational dipole moment densities associated with translations
and local Lorentz rotations.

For ease of comparison, the formulas of
Ref.~\cite{obukhov2014spin} are expressed throughout this appendix in
the notation, sign conventions, and tensor normalizations of the present
paper. Any corresponding normalization factors are absorbed into the
phenomenological coefficients \(\rho'\) and \(\tau'\), whose symbols are
retained.

The relevant Dirac dipole moments are
\begin{align}
M_\alpha{}^{ij}
&:=
\frac{i}{4m}
\left(
\bar\psi\sigma^{ij}D_\alpha\psi
-
(D_\alpha\bar\psi)\sigma^{ij}\psi
\right),
\label{eq:obukhov-translational-moment}
\\
M_{\alpha\beta}{}^{ij}
&:=
\frac{1}{8m}
\bar\psi
\left(
\sigma^{ij}\sigma_{\alpha\beta}
+
\sigma_{\alpha\beta}\sigma^{ij}
\right)\psi
\nonumber\\
&=
\frac{1}{8m}
\bar\psi\{\sigma^{ij},\sigma_{\alpha\beta}\}\psi .
\label{eq:obukhov-lorentz-moment}
\end{align}

The nonminimal modification has the symbolic Pauli form
\begin{equation}
\Delta\mathcal L_{\rm nm}
\sim
S_{ij}{}^\alpha M_\alpha{}^{ij}
+
R_{ij}{}^{\alpha\beta}M_{\alpha\beta}{}^{ij},
\end{equation}
which may be written as
\begin{align}
\Delta\mathcal L_{\rm nm}
&=
\frac{\rho'}{2}S_{ij}{}^\alpha
\left(
\bar\psi\sigma^{ij}D_\alpha\psi
-
(D_\alpha\bar\psi)\sigma^{ij}\psi
\right)
\nonumber\\
&\quad
-
\tau'\bar\psi
\left(
R+iP\gamma_5
\right)
\psi .
\label{eq:obukhov-nonminimal}
\end{align}
Here \(\rho'\) and \(\tau'\) are phenomenological coupling constants, and
\(R\) and \(P\) denote the curvature scalar and pseudoscalar. The sign of
the \(P\)-term is written in the present convention.

Varying \eqref{eq:obukhov-nonminimal} with respect to \(\bar\psi\) gives
the additional terms
\begin{align}
\Delta_\rho
&=
\rho'
\left[
S_{ij}{}^\alpha\sigma^{ij}D_\alpha\psi
+
\frac{1}{2}
(D_\alpha S_{ij}{}^\alpha)\sigma^{ij}\psi
\right],
\label{eq:obukhov-Delta-rho}
\\
\Delta_\tau
&=
-\tau'
\left(
R+iP\gamma_5
\right)\psi ,
\label{eq:obukhov-Delta-tau}
\end{align}
where \(D_\alpha\sigma^{ij}=0\) has been used. Thus the Obukhov-modified
Dirac equation is
\begin{equation}
\left(
i\gamma^\alpha D_\alpha-m
\right)\psi
+
\Delta_\rho
+
\Delta_\tau
=
0 .
\label{eq:obukhovFE-compact}
\end{equation}
Equivalently,
\begin{align}
0
&=
\left(
i\gamma^\alpha D_\alpha-m
\right)\psi
\nonumber\\
&\quad
+
\rho'
\left[
S_{ij}{}^\alpha\sigma^{ij}D_\alpha\psi
+
\frac{1}{2}
(D_\alpha S_{ij}{}^\alpha)\sigma^{ij}\psi
\right]
\nonumber\\
&\quad
-
\tau'
\left(
R+iP\gamma_5
\right)\psi .
\label{eq:obukhovFE}
\end{align}

On the other hand, in the \(\ell_5=0\) axial-torsion limit,
\eqref{eq:GenDirac-RC-E} gives
\begin{align}
0
&=
\left(
i\gamma^\alpha D_\alpha-m
\right)\psi
-
\frac{i\ell}{3}
S_{ij}{}^\alpha\sigma^{ij}D_\alpha\psi
\nonumber\\
&\quad
-
\frac{\ell}{12}
\left(
R+iP\gamma_5
+
2i\sigma^{ij}R_{[ij]}
\right)\psi .
\label{eq:ccgg-comp1}
\end{align}
For axial torsion,
\begin{equation}
R_{[ij]}=S^\alpha{}_{ij;\alpha}.
\label{eq:axialRiccit}
\end{equation}
Substituting \eqref{eq:axialRiccit} into \eqref{eq:ccgg-comp1} and
rearranging yields
\begin{align}
0
&=
\left(
i\gamma^\alpha D_\alpha-m
\right)\psi
\nonumber\\
&\quad
-
\frac{i\ell}{3}
\left[
S_{ij}{}^\alpha\sigma^{ij}D_\alpha\psi
+
\frac{1}{2}
(D_\alpha S_{ij}{}^\alpha)\sigma^{ij}\psi
\right]
\nonumber\\
&\quad
-
\frac{\ell}{12}
\left(
R+iP\gamma_5
\right)\psi .
\label{eq:ccgg-comp2}
\end{align}
Comparing \eqref{eq:obukhovFE} and \eqref{eq:ccgg-comp2}, the two equations
have the same operator structure in this axial-torsion limit, with the
parameter identification
\begin{equation}
\rho'=-\frac{i\ell}{3},
\qquad
\tau'=\frac{\ell}{12},
\end{equation}
where the imaginary term is necessarily expected from \eqref{eq:obukhov-translational-moment} and \eqref{eq:obukhov-nonminimal}.

\section{The \texorpdfstring{\(\sigma\sigma R\)}{sigma-sigma-R} Identity}
\label{app:sigmaR}

Let $R_{\mu\nu\alpha\beta}$ be antisymmetric in $(\mu\nu)$ and in $(\alpha\beta)$, and define
\begin{equation}
R_{\mu\nu}:=R^\rho{}_{\mu\rho\nu},\quad R:=g^{\mu\nu}R_{\mu\nu},\quad
P:=\frac12\,\varepsilon^{\mu\nu\alpha\beta}R_{\mu\nu\alpha\beta}.
\end{equation}
Decomposition:
\begin{equation}
\sigma^{\alpha\beta}\sigma^{\mu\nu}
=\frac12[\sigma^{\alpha\beta},\sigma^{\mu\nu}]
+\frac12\{\sigma^{\alpha\beta},\sigma^{\mu\nu}\}.
\label{eq:decomp}
\end{equation}
(Anti)commutator identities:
\begin{equation}
\begin{aligned}
{} [\sigma^{\alpha\beta},\sigma^{\mu\nu}]&=2i\Big(
g^{\beta\mu}\sigma^{\alpha\nu}
-g^{\alpha\mu}\sigma^{\beta\nu}\\
&\qquad
-g^{\beta\nu}\sigma^{\alpha\mu}
+g^{\alpha\nu}\sigma^{\beta\mu}
\Big).
\end{aligned}
\label{eq:comm}
\end{equation}
\begin{equation}
\begin{aligned}
\{\sigma^{\alpha\beta},\sigma^{\mu\nu}\}&=2\bigl(
g^{\alpha\mu}g^{\beta\nu}
-g^{\alpha\nu}g^{\beta\mu}
\bigr)\\
&\quad
+2i\varepsilon^{\alpha\beta\mu\nu}\gamma^5.
\end{aligned}
\label{eq:anticomm}
\end{equation}
Apply \eqref{eq:comm} and contract with $R_{\mu\nu\alpha\beta}$:
\begin{align}
[\sigma^{\alpha\beta},\sigma^{\mu\nu}]
R_{\mu\nu\alpha\beta}&=2i\Big(
g^{\beta\mu}\sigma^{\alpha\nu}
-g^{\alpha\mu}\sigma^{\beta\nu}
\nonumber\\
&\quad
-g^{\beta\nu}\sigma^{\alpha\mu}
+g^{\alpha\nu}\sigma^{\beta\mu}
\Big)R_{\mu\nu\alpha\beta}
\nonumber\\
&=2i\Big(
\sigma^{\alpha\nu}R^{\beta}{}_{\nu\alpha\beta}
-\sigma^{\beta\nu}R^{\alpha}{}_{\nu\alpha\beta}
\nonumber\\
&\quad
-\sigma^{\alpha\mu}R_{\mu}{}^{\nu}{}_{\alpha\nu}
+\sigma^{\beta\mu}R_{\mu}{}^{\alpha}{}_{\alpha\beta}
\Big)
\nonumber\\
&=2i\Big(
-\sigma^{\alpha\nu}R_{\nu\alpha}
-\sigma^{\beta\nu}R_{\nu\beta}
\nonumber\\
&\quad
-\sigma^{\alpha\mu}R_{\mu\alpha}
-\sigma^{\beta\mu}R_{\mu\beta}
\Big)
\nonumber\\
&=8i\,\sigma^{\mu\nu}R_{\mu\nu}.
\label{eq:commcontract}
\end{align}
In the third line we used only antisymmetry in the last pair $(\alpha\beta)$ and in the first pair $(\mu\nu)$, e.g.
\[
\begin{gathered}
R^{\beta}{}_{\nu\alpha\beta}
=-R^{\beta}{}_{\nu\beta\alpha}
=-R_{\nu\alpha},\\
R_{\mu}{}^{\nu}{}_{\alpha\nu}=R_{\mu\alpha}.
\end{gathered}
\]
Apply \eqref{eq:anticomm} and contract with $R_{\mu\nu\alpha\beta}$:
\begin{align}
\{\sigma^{\alpha\beta},\sigma^{\mu\nu}\}
R_{\mu\nu\alpha\beta}
&=\Big[
2\bigl(
g^{\alpha\mu}g^{\beta\nu}
-g^{\alpha\nu}g^{\beta\mu}
\bigr)
\nonumber\\
&\quad
+2i\varepsilon^{\alpha\beta\mu\nu}\gamma^5
\Big]R_{\mu\nu\alpha\beta}
\nonumber\\
&=2\bigl(
g^{\alpha\mu}g^{\beta\nu}
-g^{\alpha\nu}g^{\beta\mu}
\bigr)R_{\mu\nu\alpha\beta}
\nonumber\\
&\quad
+2i\,\varepsilon^{\mu\nu\alpha\beta}
R_{\mu\nu\alpha\beta}\gamma^5
\nonumber\\
&=4R+4iP\gamma^5.
\label{eq:anticommcontract}
\end{align}
Sum the two pieces using \eqref{eq:decomp}:
\begin{align}
\sigma^{\alpha\beta}\sigma^{\mu\nu}
R_{\mu\nu\alpha\beta}
&=
\frac12[\sigma^{\alpha\beta},\sigma^{\mu\nu}]
R_{\mu\nu\alpha\beta}
+\frac12\{\sigma^{\alpha\beta},\sigma^{\mu\nu}\}
R_{\mu\nu\alpha\beta}
\nonumber\\
&=
\frac12\bigl(+8i\,\sigma^{\mu\nu}R_{\mu\nu}\bigr)
+\frac12\bigl(4R+4iP\gamma^5\bigr)
\nonumber\\
&=
4i\,\sigma^{\mu\nu}R_{\mu\nu}
+2R+2iP\gamma^5.
\label{eq:final}
\end{align}

\bibliography{refs}

\end{document}